\documentclass[12pt]{iopart}
\usepackage{setstack}
\usepackage{amsthm}
\usepackage{amsfonts}
\usepackage{amssymb}
\usepackage{graphicx}
\usepackage{subfig}
\usepackage{bm}
\newcommand{\ud}{\,\textrm{d}}
\newcommand{\RR}{\mathbb{R}}

\newcommand{\bxi}{\boldsymbol{\xi}}
\newcommand{\bx}{\boldsymbol{x}}
\newcommand{\by}{\boldsymbol{y}}
\newcommand{\balpha}{\boldsymbol{\alpha}}
\newcommand{\bz}{\boldsymbol{z}}
\newcommand{\bv}{\boldsymbol{v}}
\newcommand{\bu}{\boldsymbol{u}}
\newcommand{\be}{\boldsymbol{e}}
\newcommand{\FF}[1]{{}_0\mathcal{F}_0^{(#1)}}
\newcommand{\GHG}[3]{{}_{#1}\mathcal{F}_{#2}^{(#3)}}
\newcommand{\CC}[2]{\mathcal{C}_{#1}^{(#2)}}
\newcommand{\PP}[2]{\mathcal{P}_{#1}^{(#2)}}
\newcommand{\QQ}[2]{\mathcal{Q}_{#1}^{(#2)}}
\newcommand{\bone}{\boldsymbol{1}}
\newcommand{\bzero}{\boldsymbol{0}}

\newcommand{\bA}{\boldsymbol{A}}
\newcommand{\bD}{\boldsymbol{D}}
\newcommand{\bmm}{\boldsymbol{m}}
\newcommand{\bQ}{\boldsymbol{Q}}
\newcommand{\ii}{\textrm{i}}
\newtheorem{theorem}{Theorem}
\newtheorem{lemma}[theorem]{Lemma}

\graphicspath{{figures/}}

\begin{document}

%\title[Dyson's Model and Dunkl Processes]{Dyson's Brownian Motion Model as a \\Special Case of Dunkl Processes and \\Dunkl's Intertwining Operators}
\title[Interacting Particles on the Line and Dunkl Intertwining Operator of Type $A$]{Interacting Particles on the Line and \\Dunkl Intertwining Operator of Type $A$: \\Application to the Freezing Regime}

\author{Sergio Andraus$^1$, Makoto Katori$^2$ and Seiji Miyashita$^1$}
\address{$^1$ Department of Physics, Graduate School of Science, University of Tokyo,
7-3-1 Hongo, Bunkyo-ku, Tokyo 113-0033}
\address{$^2$ Department of Physics, Graduate School of Science and Engineering, Chuo University, 1-13-27 Kasuga, Bunkyo-ku, Tokyo 112-8551}
\ead{andraus@spin.phys.s.u-tokyo.ac.jp}

\begin{abstract}

We consider a one-dimensional system of Brownian particles that repel each other through a logarithmic potential. We study two formulations for the system and the relation between them. The first, Dyson's Brownian motion model, has an interaction coupling constant determined by the parameter $\beta >0$. When $\beta=1,2$ and 4, this model can be regarded as a stochastic realization of the eigenvalue statistics of Gaussian random matrices. The second system comes from Dunkl processes, which are defined using differential-difference operators (Dunkl operators) associated with finite abstract vector sets called root systems. When the type-$A$ root system is specified, Dunkl processes constitute a one-parameter system similar to Dyson's model, with the difference that its particles interchange positions spontaneously. We prove that the type-$A$ Dunkl processes with parameter $k> 0$ starting from any symmetric initial configuration are equivalent to Dyson's model with the parameter $\beta=2k$. We focus on the intertwining operators, since they play a central role in the mathematical theory of Dunkl operators, but their general closed form is not yet known. Using the equivalence between symmetric Dunkl processes and Dyson's model, we extract the effect of the intertwining operator of type $A$ on symmetric polynomials from these processes' transition probability densities. In the strong coupling limit, the intertwining operator maps all symmetric polynomials onto a function of the sum of their variables. In this limit, Dyson's model freezes, and it becomes a deterministic process with a final configuration proportional to the roots of the Hermite polynomials multiplied by the square root of the process time, while being independent of the initial configuration.

\end{abstract}


\submitto{\JPA}

\maketitle
%\listoffigures

\section{Introduction}\label{intro}

Our object of study is a multivariate stochastic process that describes the Brownian motion of $N$ particles on the line. These particles interact repulsively with a magnitude given by the inverse of their relative distances. There are two formulations for this stochastic process. The first of the two, Dyson's Brownian motion model, was introduced as the stochastic process that the eigenvalues of a Gaussian random matrix undergo when its entries are independent Brownian motions \cite{Dyson62}. In this sense, Dyson's model is a dynamic version of Gaussian random matrices \cite{Dyson62B}, and therefore both models share the same parameter $\beta$  for each of the three Gaussian ensembles of random matrices, 1, 2 and 4 for the Gaussian orthogonal, unitary and symplectic ensembles, respectively \cite{Dyson62B,Mehta04,forrester10}. While $\beta$ is usually considered as a discrete parameter, here we will treat it as a real positive parameter. In that case, the interaction in Dyson's model is characterized by the coupling constant $\beta/2$ \cite{katoritanemura07}. Additionally, in view of the fact that the joint probability density function of the eigenvalues of Gaussian random matrices has the form of the partition function of a one-dimensional log-gas \cite{Mehta04}, $\beta$ also carries the meaning of inverse temperature. The vicious walker model formulated by Fisher \cite{fisher84} consists of $N$ random walkers who annihilate each other if they make contact; one of the most particular properties of Dyson's model is that for $\beta=2$, Dyson's model is a diffusion scaling limit of the vicious walker model with the restriction that no walkers annihilate \cite{katoritanemura02}. This restriction is equivalent to conditioning the $N$ particles never to collide, and therefore this process is also known as the non-colliding Brownian motion \cite{grabiner99,katoritanemura04}. Due to this connection, Dyson's model has found applications in polymer networks \cite{degennes,essamguttmann}, traffic models \cite{baik06}, gauge fields and nuclear physics \cite{bohigas83, bohigas85, tierz04, deharo05, forrester11}, and random growth models \cite{johansson00, johansson03}, among others. In particular, Dyson's model has been applied extensively in the Kardar-Parisi-Zhang universality class \cite{prahofer00, imamura05, sasamoto10, sasamotospohn10, amircorwinquastel11, imamura11, corwin12}, and there have been recent experimental observations of phenomena pertaining to it that confirm theoretical predictions \cite{takeuchi10, takeuchi11}. Also, because of its relationship with the vicious walker model, Dyson's model is intimately related to the mathematical areas of combinatorics, group theory and representation theory \cite{guttmann98, krattenthaler00, fulton, macdonald}. 

The second system we consider, the type-$A$ Dunkl processes, is a particular case of Dunkl processes. In general, these processes are defined using Dunkl operators, $\{T_i\}_{1\leq i\leq N}$, a set of differential-difference operators \cite{Dunkl89} that depend on a series of parameters and a finite set of vectors. This set of vectors is called root system (see \ref{RootSys}). In this paper, we will focus on the root system of type $A$ and its corresponding Dunkl operators given by \cite{dunklxu}
\begin{equation}
T_{i} f(\bx)=\frac{\partial}{\partial x_i}f(\bx)+k\sum_{\substack{j:j=1\cr j \neq i}}^N\frac{f(\bx)-f(\sigma_{ij} \bx)}{x_i-x_j}\label{DunklDefinitionA},
\end{equation}
where $k\geq 0$ is a real scalar parameter, $\bx$ is a vector in $\RR^N$, $\sigma_{ij}\bx$ denotes the vector $\bx$ with its $i$-th and $j$-th components interchanged and the function $f$ is an arbitrary differentiable scalar function. These operators are related to spatial partial derivatives by Dunkl's intertwining operator $V_k$ \cite{Dunkl91}. Consider a function $f(\bx)$ as in \eref{DunklDefinitionA}. Then, $V_k$ is defined by the following equation:
\begin{equation}
T_i V_k f(\bx)= V_k \frac{\partial}{\partial x_i}f(\bx).\label{intertwiningproperty}
\end{equation}
Also, $V_k$ is defined to have no effect on constants (i.e. $V_k 1= 1$) and when it is applied on a monomial of a given degree, it produces a homogeneous polynomial of the same total degree. Note that $V_k$ is different for each type of Dunkl operator. Through the use of $V_k$, many calculations are simplified, as it allows Dunkl operators to be treated as partial derivatives. However, no general closed form for it has been found. We will show how $V_k$ transforms the heat equation to obtain Dunkl processes as follows \cite{rosler08}. Let us denote the Laplacian operator by 
\begin{equation}
\Delta^{(x)}=\sum_{i=1}^N \frac{\partial^2}{\partial x_i^2}.
\end{equation}
The $N$-dimensional heat equation is given by
\begin{equation}
\left(\frac{\partial}{\partial t} - \frac{1}{2}\Delta^{(x)}\right)p_0(t,\by|\bx)=0,\label{heateq}
\end{equation}
where $p_0$ is the $N$-dimensional heat kernel. When regarded as a Kolmogorov backward equation with a given initial condition, the heat equation defines an $N$-dimensional Brownian motion, as it produces the transition probability density (TPD) that describes its evolution (the heat kernel). Let us apply $V_k$ on \eref{heateq}:
\begin{equation}
0=V_k\left(\frac{\partial}{\partial t} - \frac{1}{2}\Delta^{(x)}\right)p_0(t,\by|\bx)=\left(\frac{\partial}{\partial t} - \frac{1}{2}\sum_{i=1}^N T_i^2\right)V_kp_0(t,\by|\bx),\label{vkheateq}
\end{equation}
where we have used \eref{intertwiningproperty} and the fact that $V_k$ has no effect on the time variable. Note that \eref{vkheateq} has the form of a heat equation in which Dunkl operators appear instead of spatial partial derivatives. This is called the Dunkl heat equation, and it depends on the type of Dunkl operators chosen. As before, we can regard the Dunkl heat equation as a Kolmogorov backward equation and use it to define Dunkl processes. The associated TPD describes Dunkl processes' evolution. Hence, different root systems define different kinds of Dunkl processes. In the case of the type-$A$ Dunkl operators \eref{DunklDefinitionA}, the corresponding type-$A$ Dunkl processes are a one-parameter family of stochastic processes, representing $N$ particles undergoing Brownian motion and interacting repulsively through a logarithmic potential with the coupling constant $k>0$ \cite{rosler08}. However, these particles do not only repel each other, but they exchange places due to the operator $\sigma_{ij}$ in \eref{DunklDefinitionA}, which is different from Dyson's model. Dunkl processes subject to symmetric initial conditions with respect to their respective root systems are called radial Dunkl processes \cite{demni08A,demni08B}, and their properties have been an active topic of research in recent years \cite{demni09A,demni09B}. In order to emphasize their variable exchange symmetry, we will refer to the type-$A$ radial Dunkl processes as symmetric Dunkl processes.

While initially defined as a tool to study multivariate orthogonal polynomials \cite{Dunkl89}, Dunkl operators have also been used in harmonic analysis \cite{dunkl92} and stochastic processes \cite{rosler98}. In mathematical physics, they have been used for the study of the Calogero-Moser (CM) systems \cite{calogero71, vanDiejen97,forrester10}. These are quantum many-body integrable systems in which $N$ particles constrained to move in one dimension are confined by an external harmonic potential and interact with each other through a potential proportional to the square of the inverse of their relative distance. There are several versions of these systems, e.g. instead of the real line, the $N$ particles can be confined to the unit circle (the Calogero-Sutherland model \cite{hikamik96, lapointevinet96,ujinowadati96}) and can be given spin-like internal degrees of freedom \cite{hikamiwadati93}. Within our context, the most natural example is the CM system with particle exchange interaction considered by Polychronakos in \cite{polychronakos93}, which is described by the Hamiltonian
\begin{equation}
\mathcal{H}_{\textrm{XCM}}=-\frac{1}{2}\Delta^{(x)}+\frac{k^2}{2}\sum_{i=1}^{N}x_i^2+\sum_{1\leq i<j\leq N}\frac{k(k-\sigma_{ij})}{(x_i-x_j)^2}.\label{SCMsystem}
\end{equation}
Here, $k$ and $\sigma_{ij}$ are the same as in definition \eref{DunklDefinitionA}. Dunkl operators are used in this case to simplify the Hamiltonian \eref{SCMsystem} after transforming it using the function $\rme^{-kW}$, where $W(\bx)=\sum_{i=1}^Nx_i^2/2-\sum_{i<j}\log|x_i-x_j|$:
\begin{equation}
\tilde{\mathcal{H}}_{\textrm{XCM}}=-\frac{1}{k}\rme^{kW}(\mathcal{H}_{\textrm{XCM}}-E_0)\rme^{-kW}=\frac{1}{2k}\sum_{i=1}^N T_i^2 -\sum_{j=1}^N x_j\frac{\partial}{\partial x_j}.\label{transformedhamiltonian}
\end{equation}
Here, $E_0=[kN+k^2N(N-1)]/2$. Using this simplified form, one can show that this system is integrable (see \cite{forrester10}). 

After reviewing Dyson's model briefly in Section~\ref{subsectiondyson}, we will show in Section~\ref{subsectiondunklprocesses} that the TPD of symmetric Dunkl processes obeys a Kolmogorov backward equation that is identical to that of Dyson's model (Theorem~\ref{symmetricdunkllemma}). Therefore, their TPDs are equivalent and a relationship between the Dunkl parameter $k$ and the parameter $\beta$ in Dyson's model is established. This relationship provides a physical meaning to the abstract Dunkl parameter $k$, that is, $k$ can be understood as a parameter proportional to the inverse temperature. In Section~\ref{review2}, we will review the intertwining operator and its defining properties. In Section~\ref{main1theorem}, we will extract the effect of the intertwining operator on symmetric polynomials from the TPDs of Dyson's model and symmetric Dunkl processes using Theorem~\ref{symmetricdunkllemma}. This is the result of Theorem~\ref{maintheorem}. In order to gain a better understanding of this result, we will make observations on the particular case of quadratic symmetric polynomials in Section~\ref{quadraticsubsection}. We will find that in the $k\to\infty$ limit, $V_k$ turns all symmetric functions into a function of the sum of their variables. We will prove this fact in Section~\ref{inflimit} as Theorem~\ref{conjecture2}. Finally, in Section~\ref{main3} we will make use of the behaviour of $V_{k\to\infty}$ to investigate the limit $k=\beta/2\to\infty$ (strong coupling limit) of the TPD of the symmetric Dunkl process (or its equivalent Dyson's model). We will prove in Theorem~\ref{lasttheorem} that in this limit, which we call the freezing regime, this process becomes deterministic and its final configuration is proportional to a vector composed of the roots of the $N$-th Hermite polynomial multiplied by the square root of the process time. We will also find that the form of $V_{k\to\infty}$ causes the symmetric Dunkl process to be independent of its initial condition in the freezing regime.

\section{Dyson's Brownian Motion Model and Dunkl Processes}\label{review1}

\subsection{Dyson's Brownian Motion Model}\label{subsectiondyson}

Dyson's Brownian motion model \cite{Dyson62} is a stochastic process in which $N$ particles undergo Brownian motion in one dimension while interacting repulsively with each other with a strength proportional to the inverse of the distances between them. Dyson's model is described by the Kolmogorov backward equation (see, e.g. Remark 1 in \cite{katoritanemura07})
\begin{eqnarray}
\frac{\partial}{\partial t}P_\beta(t,\by|\bx)&=&\frac{1}{2}\Delta^{(x)}P_\beta(t,\by|\bx)+\frac{\beta}{2}\sum_{i=1}^N\sum_{\substack{j=1:\cr j\neq i}}^N\frac{1}{x_i-x_j}\frac{\partial}{\partial x_i}P_\beta(t,\by|\bx),\label{DysonBKE}
\end{eqnarray}
where $P_\beta(t,\by|\bx)$ is the TPD that the particles reach the positions $\by=(y_1, \ldots, y_N)$ after a time $t$, given that they started from $\bx=(x_1, \ldots, x_N)$. While we denote vectors in $\RR^N$ using boldface symbols, we will use the notation $x^2=\bx\cdot\bx$ to denote their squared norm. 

%%Discarded explanation on Random matrices and the temperature interpretation of beta%%

%The eigenvalues $\bzeta=\{\zeta_i\}_{i=1,\ldots,N}$ of the Gaussian ensembles of random matrices obey the joint probability density function (Theorem~3.3.1 of \cite{Mehta04})
%\begin{equation}
%\textrm{Prob}_{N,\beta}(\bzeta)\propto \exp(-\beta\mathcal{H})
%\end{equation}
%with
%\begin{equation}
%\mathcal{H}=\frac{1}{2}\zeta^2-\sum_{1\leq i<j\leq N}\log|\zeta_j-\zeta_i|.
%\end{equation}
%Because $\textrm{Prob}_{N,\beta}(\bx)$ has the form of the partition function of a log-gas in a harmonic potential in one dimension after integrating out the kinetic part, one can consider $\beta$ as a real positive parameter by giving it the meaning of inverse temperature in this log-gas model. In the same way, one can treat $\beta$ as a positive real number in Eq.~\eref{DysonBKE} by regarding $\beta/2$ as the coupling constant of the interaction (drift) term. Therefore, when $\beta$ is extended from the three values 1,2 and 4, it carries two different physical meanings. Loosely speaking, the inverse-temperature interpretation of $\beta$ still holds in Dyson's model: 

Note that when $\beta\to 0$, the interaction term in \eref{DysonBKE} vanishes leaving an $N$-dimensional Brownian motion, which means that the fluctuations in the system are much larger in magnitude than the interaction between particles. 
%This is the analog of a high temperature condition. 
On the other hand, when $\beta\to\infty$, the diffusion term is negligible compared to the interaction term, and so the randomness in the system disappears. In other words, the system freezes.
%, which is the analog of a zero-temperature condition. 

A consequence of the results in \cite{bakerforrester97} and \cite{rosler98} (see, e.g. \cite{demni08A}) is that, for $N$-dimensional vectors $\bx$ and $\by$ ordered so that $y_i<y_j$ and $x_i<x_j\ \textrm{for}\ i<j$, the TPD of Dyson's model is given by
\begin{eqnarray}
\fl P_\beta(t,\by|\bx)&=&\frac{N! \rme^{-(x^2+y^2)/2t}}{(2\pi t)^{N/2}}\prod_{j=1}^N\left[\frac{\Gamma(1+\beta/2)}{\Gamma(1+j\beta/2)}\right]\left|h_N\left(\frac{\by}{\sqrt t}\right)\right|^{\beta}\FF{2/\beta}\left(\frac{\bx}{\sqrt t},\frac{\by}{\sqrt t}\right),\label{TPDDyson1}
\end{eqnarray}
where $\Gamma(x)$ is the gamma function, $h_N(\by)=\prod_{1\leq i<j\leq N}(y_j-y_i)$ is the Vandermonde determinant and $\GHG{0}{0}{2/\beta}$ is the generalized hypergeometric function (see \ref{MVSFns}). In the case $\beta=2$, the non-colliding Brownian motion, Dyson's model has the following TPD as shown by Grabiner \cite{grabiner99}:
\begin{equation}
P_{\beta=2}(t,\by|\bx)=\frac{h_N(\by)}{h_N(\bx)}\det_{1\leq i,j\leq N}\left(\frac{\rme^{-(x_i-y_j)^2/2t}}{\sqrt{2\pi t}}\right).\label{noncolTPD}
\end{equation}

\subsection{Dunkl Processes}\label{subsectiondunklprocesses}

Within our present setting, we are interested in the type-$A$ Dunkl operators given by \eref{DunklDefinitionA} (see \ref{RootSys} and \ref{dunklstuff}). Using these operators, one can define Dunkl processes as follows \cite{rosler98}: consider the TPD, $p_k(t,\by|\bx)$, that solves the Dunkl heat equation with the parameter $k$
\begin{equation}\label{dunklheatequation}
\frac{\partial}{\partial t}p_k(t,\by|\bx)=\frac{1}{2}\sum_{i=1}^NT_i^2p_k(t,\by|\bx)
\end{equation}
as its Kolmogorov backward equation. Then the stochastic process that obeys $p_k(t,\by|\bx)$ as its TPD is defined as the Dunkl process of type $A$. Direct insertion of \eref{DunklDefinitionA} into \eref{dunklheatequation} yields \cite{dunklxu}
\begin{eqnarray}
\frac{\partial}{\partial t}p_k(t,\by|\bx)&=&\frac{1}{2}\Delta^{(x)}p_k(t,\by|\bx)+\sum_{i=1}^N\sum_{\substack{j=1:\cr j\neq i}}^N\frac{k}{x_i-x_j}\frac{\partial}{\partial x_i}p_k(t,\by|\bx)\nonumber\\
&&-\frac{k}{2}\sum_{i=1}^N\sum_{\substack{j=1:\cr j\neq i}}^N\frac{p_k(t,\by|\bx)-p_k(t,\by|\sigma_{ji}\bx)}{(x_j-x_i)^2}.\label{BKEqNDWeyl}
\end{eqnarray}
Reading off the terms of this Kolmogorov backward equation gives us information about Dunkl processes: the first term, a Laplacian acting on $\bx$, is a diffusion term, while the second term is a drift term that drives the process away from the surfaces defined by the equations $x_i=x_j$ for any $i\neq j$. The third term is a spontaneous jump term, which makes the process jump from $\bx$ to $\sigma_{ij}\bx$ for some $i\neq j$ so long as $p_k(t,\by|\bx)$ is not symmetric under permutations of $\bx$.

With the observation about the jump term in mind, comparing \eref{DysonBKE} and \eref{BKEqNDWeyl} reveals a clear similarity between the two Kolmogorov backward equations. Let us define the symmetric group $S_N$ as the set of all the permutations of $N$ objects with composition as its group operation. The permutation $\rho\in S_N$ of the vector $\bx$ is given by
\begin{equation}
\rho\bx=(x_{\rho(1)},\ldots,x_{\rho(N)}).
\end{equation}
We will consider the TPD of a Dunkl process with a symmetric initial condition. As mentioned in Section~\ref{intro}, we will call such processes symmetric Dunkl processes. Given the symmetric distribution
\begin{equation}
\mu_{\bx}^s(\bz)=\sum_{\rho\in S_N}\delta^N(\bz-\rho\bx),
\end{equation}
where we define $\delta^N(\bx-\by)=\prod_{j=1}^N\delta(x_j-y_j)$, we write
\begin{equation}
p^s_k(t,\by|\bx)=\int_{\RR^N}p_k(t,\by|\bx)\mu_{\bx}^s(\bz)\ud^N z=\sum_{\rho\in S_N}p_k(t,\by|\rho\bx),\label{pks}
\end{equation}
with $\ud^N z=\prod_{j=1}^N \ud z_j$. In the following theorem, we show that the Kolmogorov backward equation obeyed by the TPD of the symmetric Dunkl process, $p^s_k(t,\by|\bx)$, is identical to that of the TPD of Dyson's model, as noted by Demni \cite{demni08A}.

\begin{theorem}\label{symmetricdunkllemma}
The TPD of a symmetric Dunkl process, $p_k^s$, solves the Kolmogorov backward equation of Dyson's model with parameter
\begin{equation}\label{parameterequation}
\beta=2k.
\end{equation}
That is, the equivalence
\begin{equation}
p_k^s(t,\by|\bx)=P_{\beta=2k}(t,\by|\bx)
\end{equation}
is established.
\end{theorem}
\begin{proof}
Comparing \eref{BKEqNDWeyl} and \eref{DysonBKE}, we see that the only difference between the two is the third term on the RHS of \eref{BKEqNDWeyl}. 
Because none of the sums on the RHS of \eref{BKEqNDWeyl} change if the order of the summands is permuted, it follows that $p_k^s(t,\by|\bx)$ is also a solution of \eref{BKEqNDWeyl}. In particular, since $p_k^s(t,\by|\bx)$ is symmetric in $\bx$, the third term on the RHS of \eref{BKEqNDWeyl} vanishes, leaving us with \eref{DysonBKE} provided we set $k=\beta/2$.
\end{proof}
Note that, because $p_k(t,\by|\bx)$ is a probability density function on $\by$, its integral with respect to $\by$ over $\RR^N$ is equal to one, and so $p_k^s(t,\by|\bx)$ integrates to $N!$. We will exploit the equivalence between $P_\beta$ and $p_k^s$ in Section~\ref{main1}.

%Dunkl processes are based on the Dunkl operator \eref{DunklDefinition}, which depends on the definition of a number of mathematical objects. We first explain the reflection operator $\sigma_{\balpha}(\bx)$ and the concept of root system. The reflection operator is given by 
%\[\sigma_{\balpha}(\bx)=\bx-2\frac{\bx\cdot \balpha}{\balpha\cdot \balpha}\balpha,\]
%where $\balpha,\bx\in\RR^N$. $\sigma_{\balpha}(\bx)$ then denotes the vector obtained by reflecting $\bx$ across the plane normal to the vector $\balpha$. 

\section{The Intertwining Operator}\label{review2}

In the previous section we obtained some information about the TPD of Dunkl processes and how it is related to Dyson's model. However, this TPD can also be calculated using the intertwining operator $V_k$, \eref{intertwiningproperty} \cite{dunklxu,Dunkl91} (see \ref{dunklstuff} for details). 
%While we will treat $V_k$ from a mostly general perspective in this section, we will assume that $V_k$ is the $A$-type intertwining operator.

Dunkl originally defined $V_k$ for the case in which the function $f(\bx)$ in \eref{intertwiningproperty} is a polynomial, and proved its existence for $k\geq 0$ using a recursive construction. However, its explicit form is unknown in general, and known only in a few particular cases. From \eref{DunklDefinitionA}, one can note that when $k=0$, the intertwining operator is the identity operator.

The intertwining operator relates $N$-dimensional Brownian motion to Dunkl processes as follows. Let us consider the $N$-dimensional heat equation \eref{heateq} with its solution
\begin{equation}
p_0(t,\by|\bx)=\frac{\rme^{-(x^2+y^2)/2t}}{(2\pi t)^{N/2}}\rme^{\bx\cdot\by/t}\label{BMTPD}
\end{equation}
given the initial condition $p_0(0,\by|\bx)=\delta^N(\bx-\by)$. This solution is both the Green function of the diffusion equation and the TPD of the $N$-dimensional Brownian motion. When we apply $V_k$ on \eref{heateq}, we obtain the Dunkl heat equation. From this, we deduce that the Green function for the Dunkl heat equation, which we denote by $\Gamma_k(t,\by|\bx)$, is given by 
\begin{equation}
\Gamma_k(t,\by|\bx)=V_kp_0(t,\by|\bx).
\end{equation}

On the other hand, the Green function for the Dunkl heat equation can be calculated directly using the Dunkl transform \cite{dunkl92}, a generalization of the Fourier transform, in the same way that the Fourier transform can be used to derive \eref{BMTPD}, see \cite{rosler98}. In order to define the Dunkl transform, one needs to calculate the eigenfunction $E_k(\bx,\bxi)$ of the Dunkl operators $\{T_i\}_{i=1,\ldots,N}$ of parameter $k$ with eigenvalues $\bxi=(\xi_1,\ldots,\xi_N)$, also known as the Dunkl kernel \cite{Dunkl91}:
\begin{equation}
T_iE_k(\bx,\bxi)=\xi_{i}E_k(\bx,\bxi),\quad i=1,2,\ldots,N\label{eigvlprob},
\end{equation}
under the condition $E_k(\bzero,\bxi)=1$ with $\bzero=(0,\ldots,0)$. It can be expressed as
\begin{equation}
E_k(\bx,\bxi)=V_k \rme^{\bxi\cdot\bx}.\label{ansatz}
\end{equation}
One can prove that this function solves \Eref{eigvlprob} as follows:
%\[\frac{\partial}{\partial x_i}\phi_0(\bx,\bxi)=\xi_{i}\phi_0(\bx,\bxi),\]
%and its solution is 
%\begin{equation}
%\phi_0(\bx,\bxi)=\rme^{\bxi\cdot\bx}.\label{fourierkernel}
%\end{equation}
%This eigenfunction is used as the kernel of the $N$-dimensional Fourier transform with $\bx$ replaced by $-\ii\bx$, where $\ii=\sqrt{-1}$. With Eq.~\eref{fourierkernel} in mind, we adopt the form
%Insertion of Eq.~\eref{ansatz} into Eq.~\eref{eigvlprob} yields
\begin{equation}
%\fl 
T_iE_k(\bx,\bxi)=T_iV_k \rme^{\bxi\cdot\bx}=V_k \frac{\partial}{\partial x_i} \rme^{\bxi\cdot\bx}=\xi_{i}V_k \rme^{\bxi\cdot\bx}=\xi_{i}E_k(\bx,\bxi). \label{insertingansatz}
\end{equation}
Here the second equality follows from the definition of the intertwining operator and the third equality follows from the fact that $V_k$ is linear. In addition, when $E_k(\bx,\bxi)$ is written as a power series, each term is a homogeneous polynomial of $\bx$, whose degree is conserved by $V_k$. When the limit $\bx\to\bzero$ is taken, all homogeneous polynomials of degree larger than zero vanish and the only remaining term (the zero-order polynomial) is equal to one. 
%Therefore, $\phi_k(\bx,\bxi)$ as given in Eq.~\eref{ansatz} is the solution of Eq.~\eref{eigvlprob} with the condition $\phi_k(\bzero,\bxi)=1$.
%
%In analogy to Eq.~\eref{fourierkernel}, $\phi_k(\bx,\bxi)$ has the following properties \cite{opdam93,rosler99}: given $\bx,\by\in \RR$, 
%\begin{eqnarray}
%&\phi_k(\bx,\by)>0,&|\phi_k(\ii\bx,\by)|\leq 1,\nonumber\\
%&\phi_k(\bx,\by)=\phi_k(\by,\bx),\qquad&\phi_k(\lambda\bx,\by)=\phi_k(\bx,\lambda\by),\quad\lambda\in\CCC,
%\end{eqnarray}
%and
%\begin{equation}
%\phi_k(\rho\bx,\rho\by)=\phi_k(\bx,\by),\label{winvariance}
%\end{equation}
%with $\rho\in S_N$. When the Dunkl operators are not of the $A$ type, $\phi_k(\bx,\by)$ obeys property \eref{winvariance}, but the operator $\rho$ is given by the root system under consideration and Eq.~\eref{reflection}, and it is not an element of $S_N$ in general (see \cite{rosler08}). 
%Due to these properties, 
$E_k(\bx,\bxi)$ is sufficiently well behaved to be used as the kernel for the aforementioned Dunkl transform \cite{rosler08} by replacing $\bx$ with $-\ii\bx$. 

%For this reason, hereafter we denote this eigenfunction by
%\begin{equation}
%E_k(\bx,\bxi)\equiv\phi_k(\bx,\bxi).
%\end{equation}
%This function is known as the Dunkl kernel .

With this transform, one can calculate $\Gamma_k(t,\by|\bx)$ and the TPD of Dunkl processes in general (see \ref{dunklstuff}, \eref{TPDk} for the explicit relationship between $\Gamma_k$ and $p_k$). However, because the Dunkl kernel is contained in it, and the intertwining operator is unknown in general, this TPD is not explicitly specified. In our case, $p_k(t,\by|\bx)$ is given by
\begin{equation}
p_k(t,\by|\bx)=\prod_{j=1}^N\frac{\Gamma(1+k)}{\Gamma(1+jk)}\left|h_N\left(\frac{\by}{\sqrt t}\right)\right|^{2k}\frac{\rme^{-(x^2+y^2)/2t}}{(2\pi t)^{N/2}}E_k\left(\frac{\bx}{\sqrt{t}},\frac{\by}{\sqrt{t}}\right).\label{TPDkA}
\end{equation}
The last factor on the right is the only unknown part of this TPD. Let us note that in order to obtain $p_k$, we need the Dunkl kernel, which in turn requires us to calculate the effect of the intertwining operator on the exponential $\rme^{\bx\cdot\by}$. Conversely, if we have $p_k$, we can obtain the Dunkl kernel, and by expanding it into sums of suitably chosen polynomials, we can extract the effect of the intertwining operator on any polynomial.

\section{Extraction of $V_k$ from Dyson's model}\label{main1}

\subsection{Theorem}\label{main1theorem}

Here, we will obtain an expression for the effect of $V_k$ on symmetric polynomials. Using Theorem~\ref{symmetricdunkllemma}, we combine \eref{pks} and \eref{TPDkA}, equate the result to \eref{TPDDyson1}, and rescale $\bx$ and $\by$ by a factor of $\sqrt t$ to obtain
\begin{equation}
\sum_{\rho\in S_N}E_k\left(\rho \bx,\by\right)=N!\FF{1/k}\left(\bx,\by\right). \label{GHGequalsSDK}
\end{equation}
From \eref{GHGequalsSDK}, we extract the effect of the intertwining operator when applied to a symmetric polynomial. 

As preparation for the sequel, let us introduce the following definitions \cite{macdonald}: a partition $\tau$ of an integer $n$ is a non-negative integer vector such that its components are arranged in decreasing order while adding up to $n$. The components of $\tau$ are called parts, and its length, $l(\tau)$, is the number of non-zero parts in it.  The sum of its parts is called the modulus of the partition and is denoted by $|\tau|$. In other words, $\tau$ is an integer partition of $n$ if $\tau_i\geq \tau_j$ for $i<j$, and 
\begin{equation}
|\tau|=\sum_{i=1}^{l(\tau)} \tau_i=n.
\end{equation}
For example, $\tau=(\tau_1,\tau_2,\ldots)=(4,4,2,0,\ldots)$ is an integer partition of $n=10$ and has length 3. Additionally, the symbol $\tau !$ denotes the product $\prod_{j=1}^{l(\tau)}\tau_j !$. We will also use the natural ordering for partitions of the same integer defined by
\begin{equation}
\lambda<\tau\ \Leftrightarrow\ \lambda_1+\ldots+\lambda_j \leq \tau_1+\ldots+\tau_j\ {}^\forall j=1,\ldots,N,
\end{equation}
where $\lambda$ and $\tau$ are assumed to be different partitions.

We say that the ordered pair $(i,j)$ is in the partition $\tau$ (i.e. $(i,j)\in \tau$) when $1\leq i \leq l(\tau)$ and $1\leq j \leq \tau_i$. The number $\tau_j^\prime$ is the number of parts $\tau_l$ such that $\tau_l\geq j$. Finally, the partition formed by the $\tau_j^\prime$ is denoted by $\tau^\prime$, and it is called the conjugate partition of $\tau$.

We will make immediate use of two particular families of symmetric polynomials, the monomial symmetric functions and the Jack functions. Given an integer partition $\tau$, the monomial symmetric functions are given by
\begin{equation}\label{monomialsymmetric}
m_\tau(\bx)=\sum_\sigma \prod_{j=1}^N x_j^{\tau_{\sigma (j)}},
\end{equation}
where the sum is taken over all permutations $\sigma$ such that each monomial $\prod_{j=1}^N x_j^{\tau_{\sigma (j)}}$ is distinct. It is known that the Jack functions are part of the eigenfunctions of the periodic type-$A$ Calogero-Moser-Sutherland model \cite{BakerForrester97B}. They are also used to calculate the symmetric eigenfunctions of the type-$A$ Calogero-Moser system described by the Hamiltonian \eref{SCMsystem}. Within our context, they are useful to describe the action of $V_k$. The Jack function of parameter $1/k>0$, $\PP{\tau}{1/k}(\bx)$, is defined as the polynomial eigenfunction of the operator \cite{stanley89}
\begin{equation}
\left(\sum_{i=1}^Nx_i^2\frac{\partial^2}{\partial x_i^2}+2k\sum_{1\leq i\neq j\leq N}\frac{x_i^2}{x_i-x_j}\frac{\partial}{\partial x_i}\right)\PP{\tau}{1/k}(\bx)=E_{\tau,k}\PP{\tau}{1/k}(\bx)
\end{equation}
with eigenvalue 
\begin{equation}
E_{\tau,k}=\sum_{j=1}^N\tau_j[\tau_j-1-2k(j-1)]+|\tau|(N-1).
\end{equation}

Let $l_j^{\mu}$ represent the multiplicity of the $j$-th (distinct) part of $\mu$, where the subscript $P$ in $l_P^\mu$ refers to the number of distinct parts of $\mu$. In the cases where $l(\mu)<N$, there are $N-l(\mu)$ zero parts in $\mu$ and therefore $l_P^\mu=N-l(\mu)$ accounts for the multiplicity of zero parts in the first $N$ parts of $\mu$. For example, if $N=6$ and $\mu=(5,3,2,2)=(5,3,2,2,0,0)$, then $P=4$ and $l_1^{\mu}=1$, $l_2^{\mu}=1$, $l_3^{\mu}=2$ and $l_4^{\mu}=2$. Using this notation we define the following multinomial coefficient:
\begin{equation}
M(\mu,N)=\frac{N!}{l_1^{\mu}!\cdots l_P^{\mu}!}.\label{functionM}
\end{equation}
This function represents the number of distinct permutations of $\mu$ when it is considered as an $N$-dimensional vector. 

See \ref{MVSFns} for the representation of Jack functions in terms of the monomial symmetric functions and the definitions of the matrix $\{u_{\tau\lambda}(1/k)\}$, the generalized Pochhammer symbol $(kN)_\tau^{(1/k)}$ and the constants $c_\tau (1/k)$ and $c_\tau^\prime (1/k)$.

\begin{theorem}\label{maintheorem}
The action of the intertwining operator on a symmetric polynomial is given by the equation
\begin{equation}\label{mainresult}V_k m_\lambda(\bx)=\lambda!M(\lambda,N)\sum_{\substack{\tau:l(\tau)\leq N\cr |\tau|= |\lambda|}}\frac{c_\tau (1/k)}{c_\tau^\prime (1/k)}\frac{u_{\tau\lambda}(1/k)}{(kN)_\tau^{(1/k)}}\PP{\tau}{1/k}(\bx).
\end{equation}
That is, the intertwining operator of type $A$ maps the monomial symmetric functions $m_\lambda(\bx)$ onto a linear combination of Jack functions of parameter $1/k$, $\PP{\tau}{1/k}(\bx)$. 
\end{theorem}

\begin{proof}
We first expand the LHS of \eref{GHGequalsSDK} in terms of symmetric polynomials. For this purpose, we expand the symmetrized exponential $\sum_{\rho\in S_N}\exp(\rho\bx\cdot\by)$ in terms of the monomial symmetric functions $m_\mu(\bx)$.
\begin{eqnarray}
\sum_{\rho\in S_N}\rme^{\rho\bx\cdot\by}&=&\sum_{\rho\in S_N}\sum_{n=0}^\infty\sum_{\substack{\mu:l(\mu)\leq N\cr |\mu|=n}}\frac{1}{\mu!}\sum_{\substack{\tau\in S_N:\cr \tau(\mu)\ \textrm{distinct}}}\prod_{j=1}^N(x_{\rho(j)}y_j)^{\mu_{\tau(j)}}\nonumber\\
&=&\sum_{\mu:l(\mu)\leq N}\frac{1}{\mu!}\sum_{\substack{\tau\in S_N:\cr \tau(\mu)\ \textrm{distinct}}}\left\{\sum_{\rho\in S_N}\prod_{j=1}^Nx_{\rho(j)}^{\mu_{\tau(j)}}\right\}\prod_{j=1}^Ny_j^{\mu_{\tau(j)}}\nonumber\\
&=&\sum_{\mu:l(\mu)\leq N}\frac{1}{\mu!}\left\{\sum_{\rho^\prime\in S_N}\prod_{j^\prime=1}^Nx_{j^\prime}^{\mu_{\rho^\prime(j^\prime)}}\right\}\sum_{\substack{\tau\in S_N:\cr \tau(\mu)\ \textrm{distinct}}}\prod_{j=1}^Ny_j^{\mu_{\tau(j)}}
\end{eqnarray}
In the last line we have used the substitutions $j^\prime=\rho(j)$ and $\rho^\prime(j^\prime)=\tau[\rho^{-1}(j^\prime)]$. The last term on the right is, by definition, $m_\mu(\by)$. The term inside the braces is equal to $m_\mu(\bx)$ multiplied by the number of non-distinct permutations of $\mu$. Using \eref{functionM}, we write the above as
\begin{equation}
\sum_{\rho\in S_N}\rme^{\rho\bx\cdot\by}=\sum_{\mu:l(\mu)\leq N}\frac{N!m_\mu(\bx)m_\mu(\by)}{\mu!M(\mu,N)}.\label{SymDKExp}
\end{equation}
Applying $V_k$ on the above yields the LHS of \eref{GHGequalsSDK} expanded in terms of a sum of monomial symmetric functions. The next step is to eliminate the variable $\by$ using the orthogonality of Jack functions. Insertion of the inverse of \eref{JackP} in \eref{SymDKExp} after applying $V_k$ yields
\begin{equation}
\sum_{\rho\in S_N}E_{k}(\rho\bx,\by)=V_{k}\sum_{\mu:l(\mu)\leq N}\frac{N!m_\mu(\bx)}{\mu!M(\mu,N)}\sum_{\substack{\nu:\nu\leq\mu\cr |\nu|= |\mu|}}(u^{-1})_{\mu\nu}(1/k)\PP{\nu}{1/k}(\by).\label{SymmDunklKP}
\end{equation}
Using \eref{SymmDunklKP} and \eref{RadDunklKerJackP}, \eref{GHGequalsSDK} becomes
\begin{eqnarray}
V_{k}\sum_{\mu:l(\mu)\leq N}\frac{m_\mu(\bx)}{\mu!M(\mu,N)}\sum_{\substack{\nu:\nu\leq\mu\cr |\nu|= |\mu|}}(u^{-1})_{\mu\nu}(1/k)\PP{\nu}{1/k}(\by)&&\nonumber\\
\qquad\qquad\qquad\qquad\qquad 
=\sum_{\tau:l(\tau)\leq N}\frac{c_\tau (1/k)}{c_\tau^\prime (1/k)}\frac{\PP{\tau}{1/k}(\bx)\PP{\tau}{1/k}(\by)}{(kN)_\tau^{(1/k)}},&&
\end{eqnarray}
and from the orthogonality of Jack functions \cite{macdonald} and the linearity of $V_k$, which acts only on $\bx$, we can equate the coefficients of the same Jack functions of $\by$ to obtain
\begin{equation}
\sum_{\substack{\mu:l(\mu)\leq N\cr |\mu|=|\tau|}}\frac{(u^{-1})_{\mu\tau}(1/k)}{\mu!M(\mu,N)}V_{k}m_\mu(\bx)=\frac{c_\tau (1/k)}{c_\tau^\prime (1/k)}\frac{\PP{\tau}{1/k}(\bx)}{(kN)_\tau^{(1/k)}}.\label{tobeinverted}
\end{equation}
This relation is solved for $V_km_\lambda(\bx)$ if we apply the sum $\sum_\tau u_{\tau\lambda}(1/k)$ on both sides. From this last operation, the theorem follows.
\end{proof}

Theorem~\ref{maintheorem} is consistent with the expected results at particular values of $k$. For $k=0$, $V_k$ is reduced to the identity operator, for $k=1$ one obtains the known TPD of the non-colliding Brownian motion \eref{noncolTPD} \cite{grabiner99}, and in the case where all components of $\bx$ are equal, one recovers the known TPD of Dyson's model when all particles start from the same position \cite{forrester10,katoritanemura02}.

\subsection{Effect of $V_k$ on symmetric quadratic equations}\label{quadraticsubsection}

Although Theorem~\ref{maintheorem} is consistent with known results, it does not give an immediate understanding of the effect of $V_k$ on symmetric polynomials. The purpose of this subsection is to consider a simple case of Theorem~\ref{maintheorem} and make observations of the general properties of $V_k$ in order to obtain insights on the behaviour of symmetric Dunkl processes and Dyson's model. Let us consider the effect of the intertwining operator on the quadratic equations $m_2(\bx)=\sum_{i=1}^Nx_i^2=1$ and $m_{11}(\bx)=\sum_{1\leq i<j\leq N}x_ix_j=1$.
\numparts
\begin{eqnarray}
V_k\sum_{j=1}^N x_j^2=\frac{k+1}{kN+1}\sum_{j=1}^N x_j^2+\frac{2k}{kN+1}\sum_{1\leq i<j\leq N}x_ix_j=1\label{degree2a}\\
V_k\sum_{1\leq i<j\leq N} x_ix_j=\frac{k(N-1)}{2(kN+1)}\sum_{j=1}^Nx_j^2+\frac{k(N-1)+1}{kN+1}\sum_{1\leq i<j\leq N}x_i x_j=1 \label{degree2b}
\end{eqnarray}
\endnumparts
From the coefficients of the sums above, it is clear that when $k=0$ the polynomials on the LHS are unchanged. When $k=1$, these equations become
\numparts
\begin{eqnarray}
V_{k=1}\sum_{j=1}^N x_j^2&=&\frac{2}{N+1}\left(\sum_{j=1}^N x_j^2+\sum_{1\leq i<j\leq N}x_ix_j\right)=1,\\
V_{k=1}\sum_{1\leq i<j\leq N} x_ix_j&=&\frac{(N-1)}{2(N+1)}\sum_{j=1}^Nx_j^2+\frac{N}{N+1}\sum_{1\leq i<j\leq N}x_i x_j=1.
\end{eqnarray}
\endnumparts
\begin{figure}[!t]
  \centering
  \subfloat[$V_{k=0}\sum_{j=1}^3 x_j^2=1$]{\includegraphics[width=0.3\textwidth]{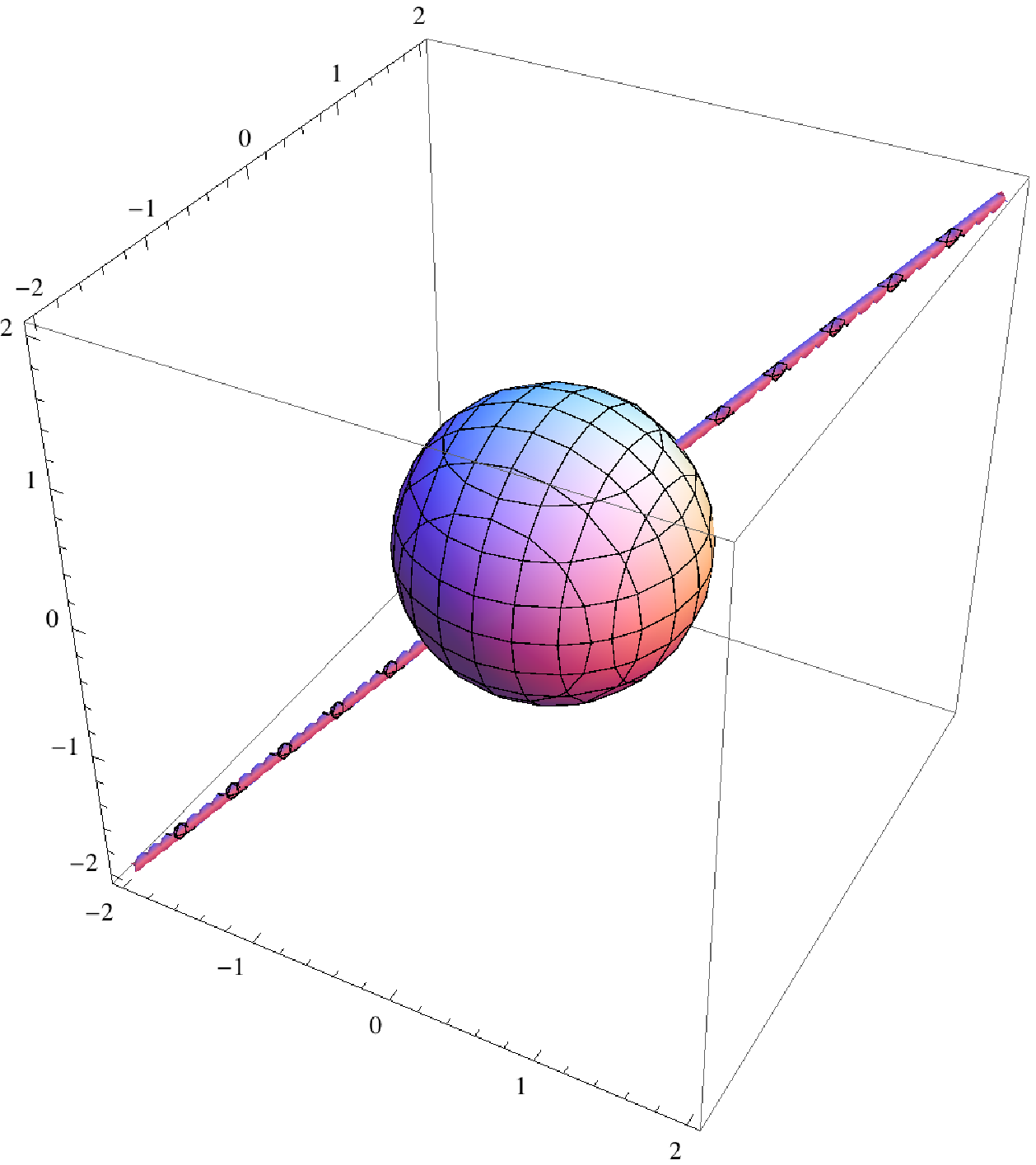}}\                 
  \subfloat[$V_{k=1}\sum_{j=1}^3 x_j^2=1$]{\includegraphics[width=0.3\textwidth]{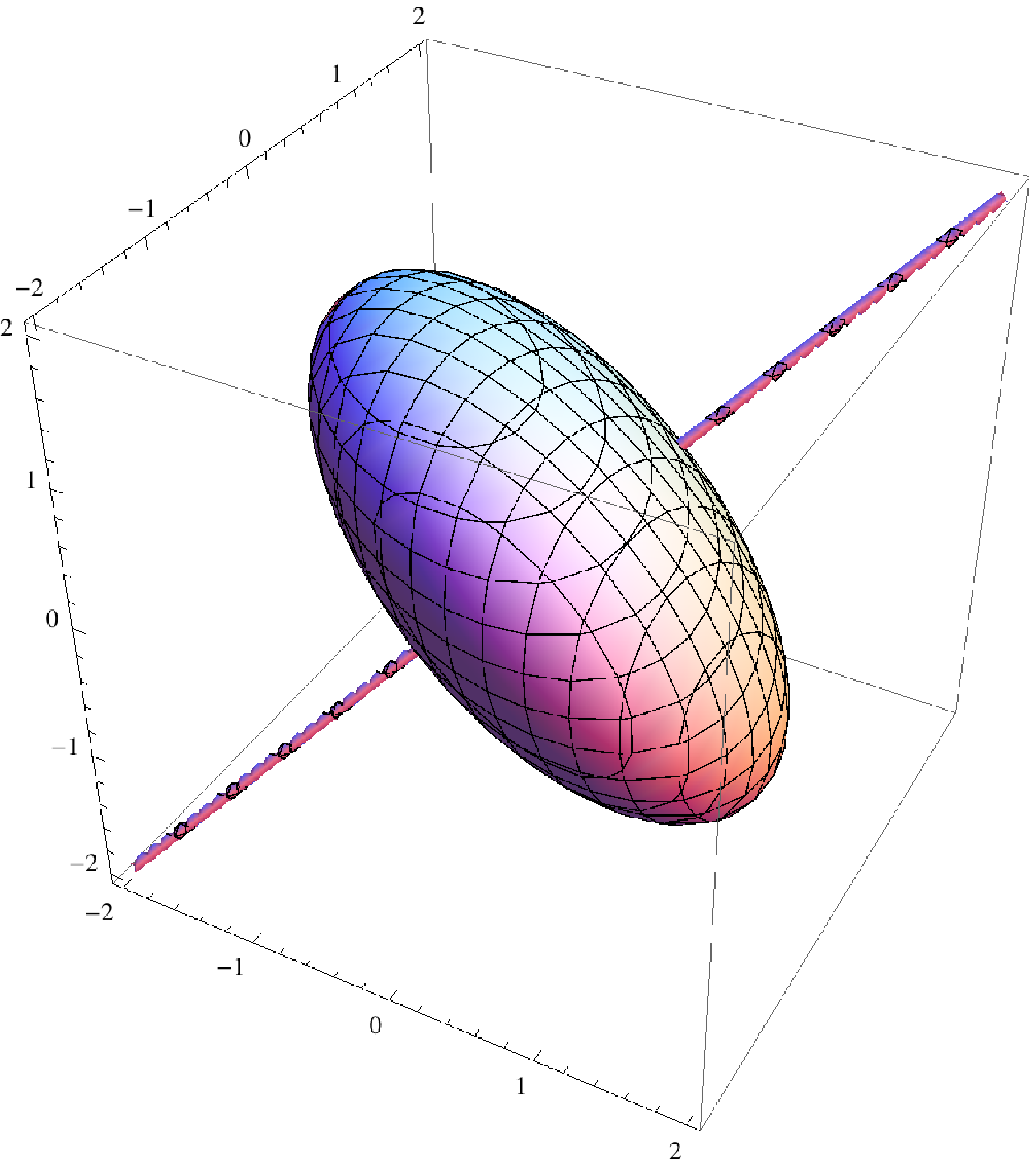}}\
  \subfloat[$V_{k\to\infty}\sum_{j=1}^3 x_j^2=1$]{\includegraphics[width=0.3\textwidth]{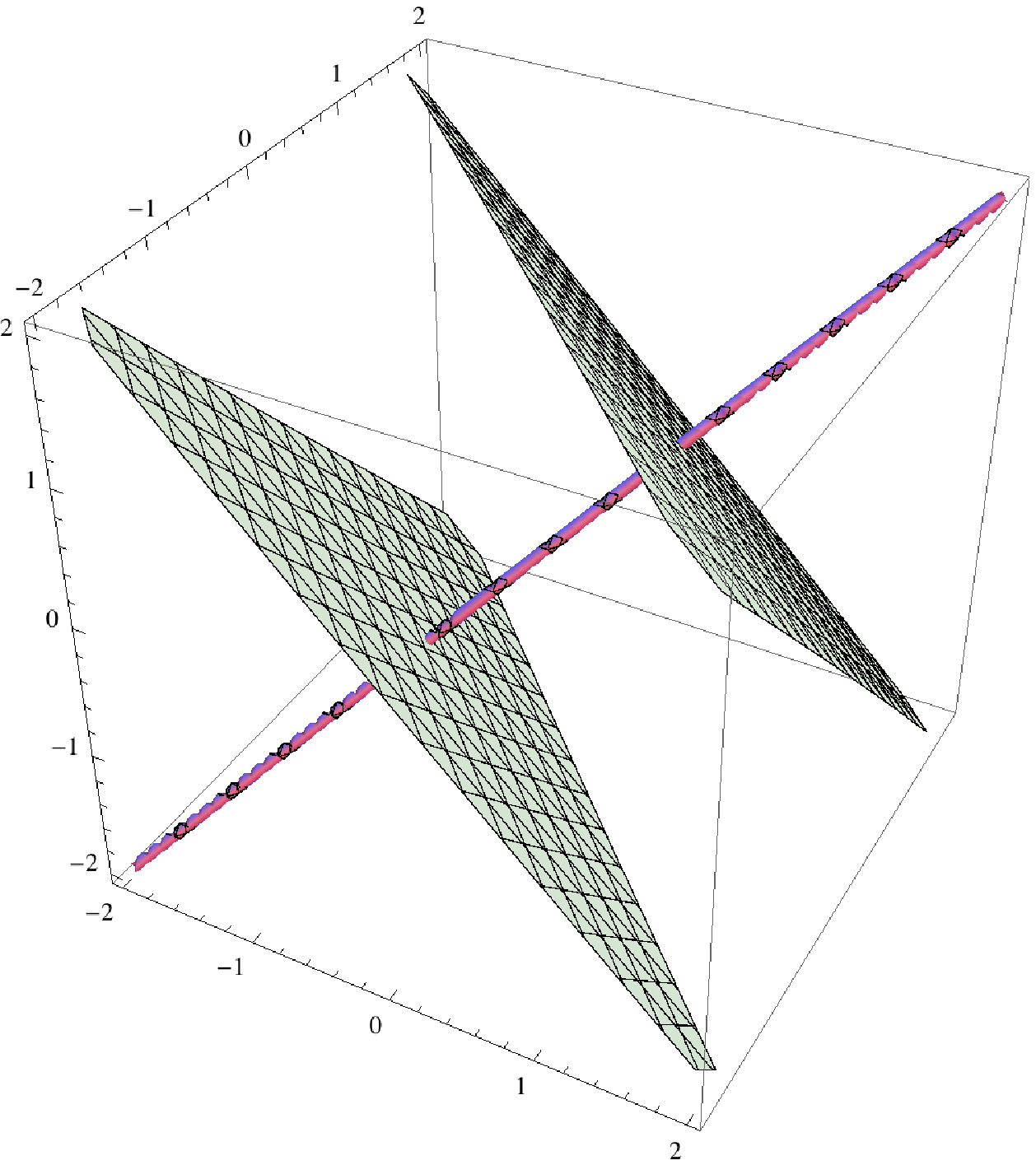}}
  \\
 
  \subfloat[$V_{k=0}\sum_{1\leq i<j\leq 3} x_ix_j=1$]{\includegraphics[width=0.3\textwidth]{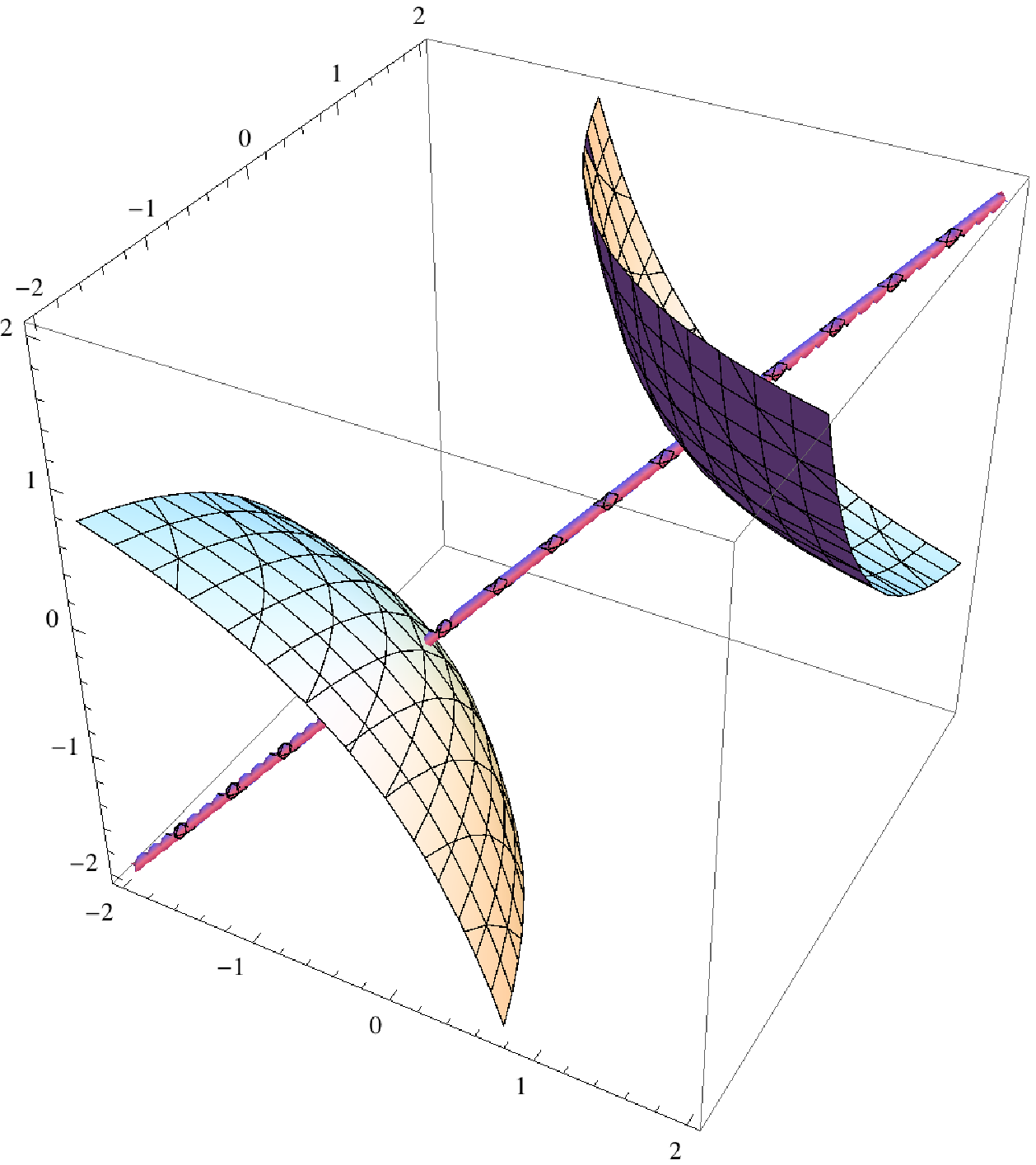}}\
  \subfloat[$V_{k=1}\sum_{1\leq i<j\leq 3} x_ix_j=1$]{\includegraphics[width=0.3\textwidth]{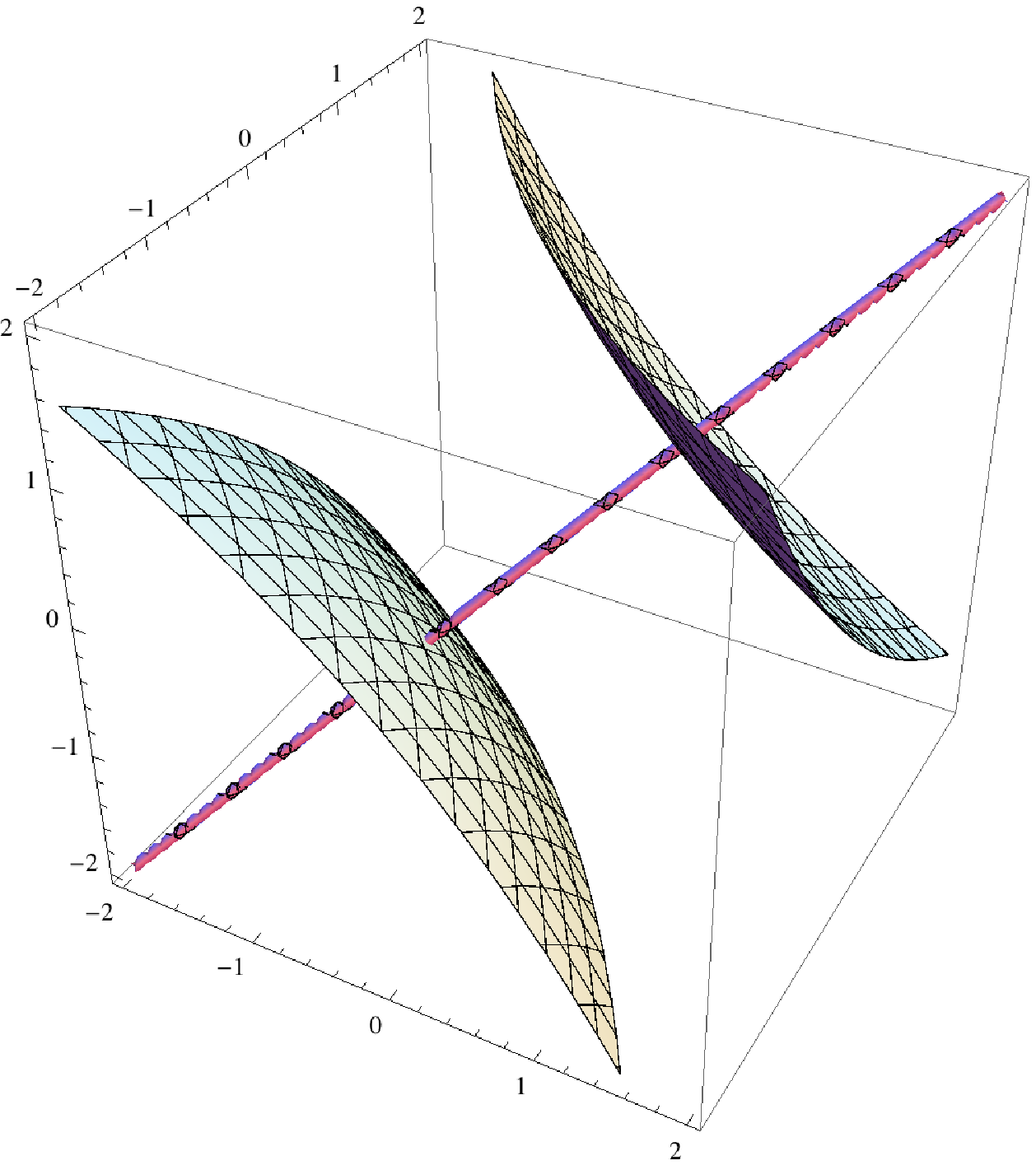}}\
  \subfloat[$V_{k\to\infty}\sum_{1\leq i<j\leq 3} x_ix_j=1$]{\includegraphics[width=0.3\textwidth]{Fig1c.eps}}
  \caption{Effect of the intertwining operator on the equations $m_2(\bx)=1$ and $m_{11}(\bx)=1$.}
  \label{quadratick1}
\end{figure}
If we take the limit $k\to\infty$, the two quadratic equations become
\numparts
\begin{eqnarray}
V_k\sum_{j=1}^N x_j^2&\stackrel{k\to\infty}{\longrightarrow}&\frac{1}{N}\left(\sum_{j=1}^N x_j\right)^2=1,\\
V_k\sum_{1\leq i<j\leq N} x_ix_j&\stackrel{k\to\infty}{\longrightarrow}&\frac{N-1}{2N}\left(\sum_{j=1}^N x_j\right)^2=1.\label{degree2kinf}
\end{eqnarray}
\endnumparts
These equations represent two pairs of planes, namely
\begin{equation}
\sum_{j=1}^N x_j=\pm\sqrt N\quad\textrm{and}\quad\sum_{j=1}^N x_j=\pm\sqrt{\frac{2N}{N-1}}.
\end{equation}
%Their distances from the origin are 1 and $2/(N-1)$ respectively, which correspond to the radius of the spheric shell $\sum_{j=1}^N x_j^2=1$ and to the shortest distance from the origin to the surface $\sum_{1\leq i<j\leq N} x_ix_j=1$. Both of these distances are achieved along the direction $\bone$.

We plot \eref{degree2a} and \eref{degree2b} for the particular case $N=3$ in Figure~\ref{quadratick1}. The line defined by the vector $(1,1,1)$ is represented by the diagonal line in each of the subfigures. From the figure, we can see that as the value of $k$ increases, the surfaces get stretched equally in all directions orthogonal to the line. In the limit $k\to\infty$, the surfaces become planes, and there are two of them because we have chosen quadratic polynomials, i.e. $|\lambda|=2$. 

From these observations, we expect that the effect of $V_{k\to\infty}$ on the monomial symmetric functions $m_\lambda$ is to reduce them to the $|\lambda|$-th power of $e_1(\bx)=\sum_{j=1}^Nx_j$. The power must be $|\lambda|$ because $V_k$ conserves the degree of all homogeneous polynomials by definition. We prove this fact in the following theorem.

\subsection{The $k\to\infty$ limit of Theorem~\ref{maintheorem}}\label{inflimit}

\begin{theorem}\label{conjecture2}
When $k\to\infty$, \eref{mainresult} becomes
\begin{equation}
\lim_{k\to\infty}V_k m_\lambda(\bx)=\frac{M(\lambda,N)}{N^{|\lambda|}}\left(\sum_{j=1}^Nx_j\right)^{|\lambda|}.
\end{equation}
\end{theorem}

\begin{proof}
The quantity
\[\frac{c_\tau(1/k)}{c_\tau^\prime (1/k)(kN)_\tau^{(1/k)}}=\prod_{(i,j)\in \tau}\frac{(\tau_i-j+k(\tau_j^\prime-i+1))}{(k(N-i+1)+j-1)(\tau_i-j+1+k(\tau_j^\prime-i))}\]
tends to zero as $k\to\infty$ in general. The only case in which this quantity does not vanish is when $\tau_j^\prime=i$ for any $(i,j)\in\tau$. Now, $\tau^\prime_j$ represents the number of parts greater than or equal to $j$ in $\tau$ (see \ref{MVSFns}), so we can conclude that this condition is satisfied only when $\tau_j^\prime=1$, i.e. $\tau=(\tau_1,0,\ldots)$. In other words, only partitions of length equal to one satisfy this condition. Therefore,
\[\frac{c_\tau(1/k)}{c_\tau^\prime (1/k)(kN)_\tau^{(1/k)}}\stackrel{k\to\infty}{\longrightarrow}0\]
whenever $l(\tau)>1$ and
\begin{eqnarray*}
\frac{c_\tau(1/k)}{c_\tau^\prime (1/k)(kN)_\tau^{(1/k)}}&=&\prod_{j=1}^{\tau_1}\frac{(\tau_1-j+k)}{(kN+j-1)(\tau_1-j+1)}\\
&&\stackrel{k\to\infty}{\longrightarrow}\prod_{j=1}^{\tau_1}\frac{1}{N(\tau_1-j+1)}=\frac{1}{N^{\tau_1}\tau_1!}
\end{eqnarray*}
when $l(\tau)=1$. Then, $V_k$ becomes
\begin{equation}
V_k m_\lambda(\bx)\stackrel{k\to\infty}{\longrightarrow}\frac{\lambda!M(\lambda,N)}{N^{|\lambda|}|\lambda|!}a_{\lambda^*\lambda}(e_{1}(\bx))^{|\lambda|},\nonumber
\end{equation}
where $\lambda^*=(|\lambda|,0,\ldots)$. This is due to the fact that if $l(\tau)=1$, then $\tau^\prime=(1,1,\ldots,1)$, a partition composed of ones with $|\lambda|$ parts (see \eref{edefinition} and Table~\ref{Jacktable} for the definitions of $e_{1}(\bx)$ and $a_{\tau\lambda}$). We finish our calculation giving an explicit form of the matrix components $a_{\lambda^*\lambda}$ by expanding $(e_{1}(\bx))^{|\lambda|}$ in terms of $m_\tau(\bx)$. We write
\[(e_{1}(\bx))^{|\lambda|}=\left(\sum_{j=1}^Nx_j\right)^{|\lambda|}=\sum_{\substack{\tau:l(\tau)\leq N\cr |\tau|=|\lambda|}}\frac{|\lambda|!}{\tau!}m_\tau(\bx)=\sum_{\substack{\tau:l(\tau)\leq N\cr |\tau|=|\lambda|}}a_{\lambda^*\tau}m_\tau(\bx),\]
and deduce that
\begin{equation}
a_{\lambda^*\lambda}=\frac{|\lambda|!}{\lambda!},
\end{equation}
which proves the statement.
\end{proof}

Note that the equation
\begin{equation}
\frac{M(\lambda,N)}{N^{|\lambda|}}\left(\sum_{j=1}^Nx_j\right)^{|\lambda|}=1
\end{equation}
clearly represents one (for $|\lambda|$ odd) or two (for $|\lambda|$ even) planes with normal vector equal to
\begin{equation}
\bone=(1,\ldots,1).\label{vectorone}
\end{equation}
This is consistent with Figures~\ref{quadratick1}(c) and (f).

\section{The Freezing Regime}\label{main3}

The physical implications of the observations in the previous section involve the behaviour of the symmetric Dunkl process when $k$ tends to infinity. This limit is equivalent to the strong coupling limit ($\beta\to\infty$) of Dyson's model. Although the observations in the previous section suggest that the limit $\lim_{k\to\infty}\sum_{\rho\in S_N}E_k(\rho\bx,\by)$ exists, one must be careful when taking this limit in \eref{TPDDyson1}. In particular, the fact that $k=\beta/2$ is the coupling constant of the interaction between particles in Dyson's model implies that this interaction separates the particles infinitely fast when $k\to\infty$. In order to obtain an intuitive picture of what occurs, we consider the Kolmogorov forward equation associated to \eref{DysonBKE} with the parameter $k$ in accordance with Theorem~\ref{symmetricdunkllemma}:
\begin{eqnarray}
\fl\frac{\partial}{\partial t}p_k^s(t,\by|\bx)=&&\frac{1}{2}\Delta^{(y)}p_k^s(t,\by|\bx)\nonumber\\
&&\qquad-\sum_{i=1}^N\sum_{\substack{j=1\cr j\neq i}}^N\frac{k}{y_i-y_j}\frac{\partial}{\partial y_i}p_k^s(t,\by|\bx)+k\sum_{i=1}^N\sum_{\substack{j=1\cr j\neq i}}^N\frac{p_k^s(t,\by|\bx)}{(y_i-y_j)^2}\label{SymFPE}.
\end{eqnarray}
We will now apply the scaling $\by\to\sqrt k \bv$. Under this scaling, \eref{SymFPE} becomes
\begin{eqnarray}
\fl\frac{\partial}{\partial t}p_k^s(t,\sqrt k \bv|\bx)=&&\frac{1}{2k}\Delta^{(v)}p_k^s(t,\sqrt k \bv|\bx)\nonumber\\
&&\quad-\sum_{i=1}^N\sum_{\substack{j=1\cr j\neq i}}^N\frac{1}{v_i-v_j}\frac{\partial}{\partial v_i}p_k^s(t,\sqrt k \bv|\bx)+\sum_{i=1}^N\sum_{\substack{j=1\cr j\neq i}}^N\frac{p_k^s(t,\sqrt k \bv|\bx)}{(v_i-v_j)^2}\label{ScaledFPE}.
\end{eqnarray}
Furthermore, if we assume that $\lim_{k\to\infty}p_k^s(t,\sqrt k \bv|\bx)k^{N/2}=p_\infty^s(t,\bv|\bx)$, taking the limit $k\to\infty$ transforms \eref{ScaledFPE} into
\begin{equation}
\frac{\partial}{\partial t}p_\infty^s(t,\bv|\bx)=-\sum_{i=1}^N\frac{\partial}{\partial v_i}\left[\sum_{\substack{j=1\cr j\neq i}}^N\frac{p_\infty^s(t,\bv|\bx)}{v_i-v_j}\right],\label{diffusionless}
\end{equation}
a Kolmogorov forward equation with a drift term but without a diffusion term. This equation suggests that the scaled process on $\bv$ is actually deterministic. Physically speaking, this is a natural consequence of taking the $k\to\infty$ limit, since the randomness of Dyson's model is removed in the freezing regime. Therefore, the argument above implies that $p_k^s(t,\sqrt k \bv|\bx)k^{N/2}$ should converge to a function of the form
\begin{equation}
\lim_{k\to\infty}p_k^s(t,\sqrt k \bv|\bx)k^{N/2}=\delta^N[\bv-\boldsymbol{F}(t,\bx)],
\end{equation}
where $\boldsymbol{F}(t,\bx)$ denotes the vector solution of a deterministic equation of motion.

Interestingly, the evolution of the scaled process considered above is related to the Hermite polynomials defined by (see e.g. \cite{arfken})
\begin{equation}
H_N(x)=(-1)^N \rme^{x^2}\frac{\ud^N}{\ud x^N}(\rme^{-x^2}).\label{hermitep}
\end{equation}
More specifically, the process described by \eref{diffusionless} depends on the $N$ roots of $H_N(x)$, here denoted by 
\begin{equation}
\bz_N=(z_{1,N},\ldots,z_{N,N}),\label{zroots}
\end{equation}
and their permutations $\rho \bz_N$, $\rho\in S_N$, where $z_{i,N}$ is the $i$-th root and $z_{i,N}<z_{j,N}$ for $i<j$. These $N$ roots are known to be real \cite{szego}. Also, because $H_N(x)$ obeys the relation
\begin{equation}
H_N(-x)=(-1)^NH_N(x),
\end{equation}
its roots add up to zero, i.e.
\begin{equation}
\sum_{j=1}^Nz_{j,N}=0.
\end{equation}

We address the freezing regime in the following theorem.

\begin{theorem}\label{lasttheorem}
The $k\to\infty$ limit of the TPD $p_k^s(t,\sqrt k \bv|\bx)k^{N/2}$ is independent of $\bx$ and is given by
\begin{equation}
p_\infty^s(t,\bv|\bx)=p_\infty^s(t,\bv)=\sum_{\rho\in S_N} \delta^N[\bv-\sqrt{2t}\rho\bz_N].
\end{equation}
\end{theorem}

\begin{proof}
We first calculate the limit
\begin{equation}
\lim_{k\to\infty}\sum_{\rho\in S_N}E_k(\rho\bx,\by).
\end{equation}
Using Theorems~\ref{maintheorem} and \ref{conjecture2} and the expansion in \eref{SymDKExp}, we obtain
\begin{eqnarray}
\lim_{k\to\infty}\sum_{\rho\in S_N}E_{k}(\rho\bx,\by)&=&\sum_{\mu:l(\mu)\leq N}\frac{N!m_\mu(\by)}{\mu!M(\mu,N)}\lim_{k\to\infty}V_{k}m_\mu(\bx)\nonumber\\
&=&\sum_{\mu:l(\mu)\leq N}\frac{N!m_\mu(\by)}{\mu!N^{|\mu|}}\left(\sum_{j=1}^Nx_j\right)^{|\mu|}\nonumber\\
&=&\sum_{\mu:l(\mu)\leq N}\frac{N!}{\mu!}m_\mu\left(\frac{(\bx\cdot\bone)\by}{N}\right)\nonumber\\
&=&N!\exp\left(\frac{(\bx\cdot\bone)(\by\cdot\bone)}{N}\right)\label{infsymDK}.
\end{eqnarray}

We now turn our attention to $p_k^s(t,\by|\bx)$. The probability that the symmetric Dunkl process reaches an infinitesimal volume $\ud ^N y$ around $\by$ from $\bx$ after a time $t$ is $p_k^s(t,\by|\bx)\ud ^N y$. This probability, after the scaling $\by\to\sqrt k \bv$ equals
\begin{eqnarray}
p_k^s(t,\by|\bx)\ud ^N y&=&p_k^s(t,\sqrt k \bv|\bx)k^{N/2}\ud ^N v\nonumber\\
&=&\exp \big\{ \log \big(p_k^s(t,\sqrt k \bv|\bx)k^{N/2}\big)\big\} \ud ^N v.\label{scaledprobdens}
\end{eqnarray}
When $k$ is sufficiently large one can use Stirling's approximation \cite{feller} combined with \eref{pks}, \eref{TPDkA} and \eref{infsymDK} to approximate the probability density in \eref{scaledprobdens} as the exponential of
\begin{eqnarray}
&&k\left[\frac{N}{2}(N-1)(1-\log t)-\sum_{j=1}^Nj\log j+2\log|h_N(\bv)|-\frac{v^2}{2t}\right]\nonumber\\
&&\ +\frac{\sqrt k}{N}(\bx\cdot\bone)(\bv\cdot\bone)+\frac{N}{2}\log k+\log(N!)-\frac{N}{2}\log(2\pi t)-\frac{x^2}{2t}.\label{extreme}
\end{eqnarray}
Let us denote the term in the square brackets by $F_N(\bv,t)$. According to this expression, in the limit $k\to\infty$ the probability density goes to infinity whenever $F_N(\bv,t)$ is positive, and vanishes when it is negative. However, this distinction is difficult to make due to the presence of the Vandermonde determinant of $\bv$ in $F_N(\bv,t)$. We show in \ref{extrema} that $F_N(\bv,t)$ is negative except in the points
\begin{equation}
\bv=\sqrt{2t} \rho \bz_N.
\end{equation}
The consequence of this is that the probability density \eref{scaledprobdens} vanishes at $\bv\neq\sqrt{2t}\bz_N$ or its permutations. On the other hand, when $\bv=\sqrt{2t}\bz_N$ or its permutations, $kF_N(\bv,t)$ vanishes (see \ref{extrema}) along with the term proportional to $\sqrt{k}$, because the roots of $H_N(x)$ add up to zero. That is,
\begin{equation}
\frac{\sqrt k}{N}(\bx\cdot\bone)(\sqrt{2t}\bz_N\cdot\bone)=\frac{\sqrt{2kt}}{N}(\bx\cdot\bone)\sum_{i=1}^Nz_{i,N}=0.
\end{equation}
The only $k$-dependent term that remains is $\frac{N}{2}\log(k)$, which tends to infinity with $k$. Therefore, we can write
\begin{equation}
\lim_{k\to\infty}p_k^s(t,\sqrt k \bv|\bx)k^{N/2}\ud ^N v \propto \sum_{\rho\in S_N}\delta[\bv-\sqrt{2t}\rho\bz_N].
\end{equation}
However, we know that $p_k^s(t,\sqrt k \bv|\bx)k^{N/2}\ud ^N v$ is equal to $N!$ when integrated over $\RR^N$, for any value of $k$. Similarly, the integral over $\RR^N$ of the sum of delta functions on the RHS is also equal to $N!$. Therefore, the two quantities above are equal.
\end{proof}

Note that the effect of $\lim_{k\to\infty}\sum_{\rho\in S_N}E_{k}(\rho\bx,\by)$ on the final result is that it makes the TPD $p_\infty^s(t,\bv|\bx)$ independent of $\bx$. When $F_N(\bv,t)$ attains its maximum value, the symmetric Dunkl kernel is equal to one for all values of $\bx$, which leaves $k^{N/2}$ as the dominating term in the limit. Therefore, the form of the intertwining operator and the properties of the Hermite polynomials are ultimately responsible for removing the dependence of $p_k^s$ on $\bx$ in the freezing regime.

%%%%%%%%%%%%%%%%%%%%%%%%%%%%%%%%%%%%%%%%%%%
%%%%%%%%%%%%%%%%%%%%%%%%%%%%%%%%%%%%%%%%%%%
%%%%%%%%%%%%%%%%%%%%%%%%%%%%%%%%%%%%%%%%%%%
%%%%%%%%%%%%%%%%%%%%%%%%%%%%%%%%%%%%%%%%%%%
%%%%%%%%%%%%%%%%%%%%%%%%%%%%%%%%%%%%%%%%%%%
%%%%%%%%%%%%%%%%%%%%%%%%%%%%%%%%%%%%%%%%%%%
%%%%%%%%%%%%%%%%%%%%%%%%%%%%%%%%%%%%%%%%%%%
%%%%%%%%%%%%%%%%%%%%%%%%%%%%%%%%%%%%%%%%%%%
%%%%%%%%%%%%%%%%%%%%%%%%%%%%%%%%%%%%%%%%%%%
%%%%%%%%%%%%%%%%%%%%%%%%%%%%%%%%%%%%%%%%%%%

\section{Concluding Remarks}\label{conclusions}

We established the equivalence between Dyson's model and symmetric Dunkl processes under the parameter relation \eref{parameterequation} in Theorem~\ref{symmetricdunkllemma}. This correspondence yields a physical interpretation of the abstract parameter $k$, i.e. one can regard $k$ as a coupling constant or a parameter proportional to the inverse temperature. We have also extracted the effect of the intertwining operator on symmetric polynomials in Theorem~\ref{maintheorem}. It is of interest to note that the functions involved, which are commonly found in combinatorics and representation theory \cite{macdonald, fulton} and in the Calogero-Moser systems \cite{hikamik96, lapointevinet96, dunkl02}, are found here through the comparison of two multivariate stochastic processes. This result stresses that symmetric polynomials play an important role in the setting of stochastic processes. While it is known that Schur functions are essential in the study of the non-colliding Brownian motion ($\beta/2=k=1$) \cite{katoritanemura02, katoritanemura04, guttmann98, krattenthaler00}, their role in the general $k$ case is played by Jack functions \cite{bakerforrester97, forrester10}, which makes them equally essential for the study of Dyson's model. 
%Additionally, Jack polynomials themselves are a particular case of Macdonald polynomials \cite{macdonald}, and in view of Borodin and Corwin's work \cite{borodincorwin11}, there is a possibility that $A$-type Dunkl processes are a particular case of Macdonald processes. 

We showed in Theorem~\ref{conjecture2} that the limit $k\to\infty$ of $V_k$ for symmetric polynomials exists and that it has a simple form. This fact allowed us to calculate the strong coupling limit of the symmetrized Dunkl kernel, and in turn obtain the TPD of the symmetric Dunkl process, which is the result of Theorem~\ref{lasttheorem}. The positions of the Brownian particles are locked to a deterministic path, which is a natural consequence of the small scale of the Brownian vibrations in this regime. Moreover, we pointed out two non-trivial properties of the symmetric Dunkl process in the freezing regime. The first of them is that the final configuration of the Brownian particles depends on the roots of the Hermite polynomials rather than being an equally-spaced or simpler configuration. This is a consequence of the form of the interaction between particles: the Vandermonde determinant of $\bv$ appears in $p_k^s(t,\sqrt k \bv|\bx)$ as a death factor that makes the TPD vanish whenever $v_i=v_j,\ i\neq j$, pushing the particles away from each other. Furthermore, the partial derivatives of the Vandermonde determinant are responsible for the form of the drift term in \eref{DysonBKE} and \eref{BKEqNDWeyl}. Therefore, a different interaction would lead to a different form for the function $F_N(\bv,t)$, which would in turn have its maxima in points different from $\bz_N$ or its permutations. The second property is that $\lim_{k\to\infty}p_k^s(t,\sqrt k \bv|\bx)k^{N/2}$ is independent of the initial configuration $\bx$. Because the process is deterministic in this limit, one expects it to depend on the initial conditions from which it has evolved. Contrary to this expectation, the nature of the intertwining operator at $k\to\infty$ and the symmetry that the roots of $H_N(x)$ obey take away that dependence.

We find that there is a similarity between the strong coupling limits of the symmetric Dunkl processes and the Calogero-Moser system described by \eref{SCMsystem}. When $k\to\infty$, the particles freeze at the minima of the potential in the leading ($k^2$) term
\begin{equation}
\frac{1}{2}\sum_{i=1}^{N}x_i^2+\sum_{1\leq i<j\leq N}\frac{1}{(x_i-x_j)^2},
\end{equation}
which are known to be the roots of $H_N(x)$, $\bz_N$ \cite{frahm93}. As the particles freeze, the kinetic term becomes negligible, and the only remaining variable term in \eref{SCMsystem} is the term of order $k$. This is called the ``freezing trick'' \cite{barbafinkel08}. What remains is known as the Polychronakos-Frahm (PF) spin chain Hamiltonian \cite{polychronakos93,frahm93}, and it is given by 
\begin{equation}\label{PFsc}
\mathcal{H}_{\textrm{PF}}=\sum_{1\leq i<j\leq N}\frac{\sigma_{ij}}{(z_{i,N}-z_{j,N})^2}.
\end{equation}
In this case as well as in the symmetric Dunkl processes, the particles freeze in positions related to the roots of the Hermite polynomials. This fact seems to stem from the mathematical structure of the two models, because the expressions that define them, equations \eref{vkheateq} and \eref{transformedhamiltonian}, depend on the Dunkl operators $T_i$. We suspect that there must be an underlying physical reason for this similarity, but this is a conjecture that we leave for future study.

%In summary, we have succeeded in extracting an expression for the intertwining operator, we have observed its effect on symmetric polynomials and we have used it to extract information about Dyson's model and its behaviour in the strong coupling limit. As mentioned in Secs.~\ref{intro} and \ref{subsectiondyson}, random matrix theory and its stochastic version, Dyson's model, have been thoroughly studied and applied in numerous branches of physics and mathematics. Therefore, in view of the present results we expect Dunkl processes and Dunkl theory to be applied more broadly in physics in general, and in non-equilibrium statistical mechanics in particular.

While Theorem~\ref{maintheorem} allowed us to carry out the analysis of the strong coupling limit of symmetric Dunkl processes, the expression in \eref{mainresult} should yield more information if it is simplified, e.g. by considering the eigenfunctions and eigenvalues of $V_k$. For instance, if we define the quantities
\begin{eqnarray}
B_\tau(k)=\sqrt{\frac{c_\tau (1/k)}{c_\tau^\prime (1/k)(kN)_\tau^{(1/k)}}}&\quad\textrm{and}\quad&\QQ{\tau}{k}(\bx)=B_\tau(k)\PP{\tau}{1/k}(\bx),
\end{eqnarray}
we can rewrite \eref{mainresult} in two ways in matrix form as
\begin{eqnarray}
V_k \bmm = \bD \bA \bA^T \bmm&\quad\textrm{and}\quad&V_k \bQ = \bA^T \bD \bA \bQ,
\end{eqnarray}
where $(\bmm)_\lambda=m_\lambda(\bx)$ and $(\bQ)_\tau=\QQ{\tau}{k}(\bx)$ are vectors with polynomial components, and the matrices that relate them are given by $(\bA)_{\mu\tau}=u_{\tau\mu}(1/k)B_\tau(k)$ and $(\bD)_{\lambda\mu}=\lambda!M(\lambda,N)\delta_{\lambda\mu}$. With these expressions, it should be possible to find a more revealing form for $V_k$.

The non-symmetric case of type-$A$ Dunkl processes is of interest due to the richness brought on by the spontaneous jump term (the third term on the RHS of \eref{BKEqNDWeyl}). In analogy with \eref{mainresult}, we expect to be able to express $V_k$ in terms of non-symmetric Jack functions \cite{forrester10}. Brute-force calculations in particular cases yield
\begin{equation}
V_kx_i=\frac{x_i+k\bx\cdot\bone}{1+Nk}
\end{equation}
in the linear case for general $N$ and
\numparts
\begin{eqnarray}
V_k x_i^2&=&\frac{2x_i^2+k(\bx\cdot\bone)^2}{2(1+2k)}\quad(N=2),\\
V_k x_i^2&=&\frac{2x_i(x_i+k\bx\cdot\bone)+k(x^2+k(\bx\cdot\bone)^2)}{(2+3k)(1+3k)}\quad(N=3)
\end{eqnarray}
\endnumparts
for the quadratic case. All these equations point toward a relation of the form
\begin{equation}
\lim_{k\to\infty} V_k \bx^\tau\propto (\bx\cdot\bone)^{|\tau|}.
\end{equation}
If it is correct, the above relationship could yield a Dunkl kernel of a form similar to \eref{infsymDK} in the strong coupling limit. These two statements for the non-symmetric case remain to be tested.

Dumitriu and Edelman \cite{DumitriuEdelman05} have carried out a freezing regime analysis of the eigenvalues of the tridiagonal random matrix $\beta$-ensembles. Their study includes not only the eigenvalue behaviour at $\beta\to\infty$, but also the first-order corrections at large finite values of $\beta$ (Theorem 3.1). It can be shown, using \eref{mainresult} combined with \eref{TPDkA}, that the statistics of the final positions of symmetric Dunkl processes starting from the initial condition $\bx=\bzero$ at $t=0$ coincides with the eigenvalue statistics of the $\beta$ ensembles with $\beta=2k$. Therefore, we expect to obtain the first-order correction for large finite values of $k$ for the present symmetric Dunkl processes as well. However, due to the nature of our formulation, we will have to investigate the first order ($O(1/k)$) correction of $V_k$. We intend to address this calculation in the near future.

Finally, the existence of explicit forms for the symmetrized Dunkl kernel for the type-$B$ and type-$D$ root systems, as given in \cite{bakerforrester97, demni08C}, should allow for an analysis similar to that in this work on other types of Dunkl processes. In particular, we expect the radial type-$B$ Dunkl process to be related with the type-$B$ PF spin chain \cite{yamamototsuchiya96} in a manner similar to the symmetric Dunkl processes and the PF spin chain described by \eref{PFsc}.

%%%%%%%%%%%%%%%%%%%%%%%%%%%%%%%%%%%%%%%%%%%%%%%%
%\begin{acknowledgments}
\ack{
%%%%%%%%%%%%%%%%%%%%%%%%%%%%%%%%%%%%%%%%%%%%%%%%
The authors would like to thank S. Kakei for his remarks on the CM systems and the PF spin chain, and the referees for their helpful comments and suggestions. 
SA would like to thank T. Mori for his helpful insights on this work. 
SA is supported by the Monbukagakusho: MEXT scholarship for research students.
MK is supported by
the Grant-in-Aid for Scientific Research (C)
(Grant No.21540397) from the Japan Society for
the Promotion of Science.
}
%\end{acknowledgments}
%%%%%%%%%%%%%%%%%%%%%%%%%%%%%%%%%%%%%%%%%%%%%%%%

%

\appendix

\section{Root Systems and Multiplicity Functions}\label{RootSys}

The definition of a root system depends on the reflection operator \cite{rosler08}, which is given by 
\begin{equation}\label{reflection}
\sigma_{\balpha}\bx=\bx-2\frac{\bx\cdot \balpha}{\balpha\cdot \balpha}\balpha,
\end{equation}
where $\balpha,\bx\in\RR^N$, and $\sigma_{\balpha}\bx$ denotes the vector obtained by reflecting $\bx$ across the plane normal to the vector $\balpha$.

A root system is defined as a finite set of vectors $R$ which is invariant under reflections along its own elements, called roots. That is, $\sigma_{\balpha}\bxi\in R$ for all $\balpha,\bxi\in R$. We will assume that $R$ is reduced, i.e. for any one of its elements $\bxi$, $b \bxi\in R$ implies $b=\pm 1$. It is assumed that none of the roots is a zero vector.

For every root system, a base can be chosen such that all roots are a linear combination of the base vectors with either all coefficients negative or all coefficients positive. The elements of such a base are called simple roots, and they divide the root system into the positive and negative subsystems, here denoted as $R_+$ and $R_-$. The choice of simple roots is not unique, but the positive and negative subsystems they define always divide the root system into two disjoint parts of the same size (cardinality).

The multiplicity function $k$ is a function of a root in $R$ that returns a real value and that is invariant under reflections along elements of $R$. That is,
\begin{equation}
k(\balpha)=k(\sigma_{\balpha^\prime}\balpha)\ {}^\forall \balpha, \balpha^\prime \in R.
\end{equation}
The multiplicity function can be understood as a set of parameters (or multiplicities) which are assigned to each subset of $R$ formed by roots that are related by a reflection (or composition of reflections) along some other root(s).

It is a well-known fact \cite{rosler08} that the multiplicity function reduces to a single parameter for the $A_{N-1}$ root system. For simplicity, we will denote this root system by $A$. Let us prove this fact as follows. $A$ is given by
\begin{equation}
A=\{\balpha_{ij}=(\be_i-\be_j)/\sqrt{2}:i,j=1,\ldots,N;i\neq j\}.
\end{equation}
We will choose the positive subsystem
\begin{equation}
A_{+}=\{\balpha_{ij}=(\be_i-\be_j)/\sqrt{2}:1\leq j<i\leq N\},\label{positivesubsystem}
\end{equation}
where $\be_i$ denotes the $i$-th unit base vector. This positive subsystem is generated by the simple roots
\begin{equation}
\balpha_{i+1,i}=(\be_{i+1}-\be_i)/\sqrt{2},\ i=1,\ldots,N-1.\label{simpleroots}
\end{equation}

Hereafter, we will write $\sigma_{\alpha_{ij}}=\sigma_{ij}$. We must note that the effect of $\sigma_{ij}$ on an arbitrary vector $\bx$ is that of exchanging its $i$-th and $j$-th components. To see this, we will compute the $l$-th component of $\sigma_{ij}\bx$:
\begin{equation}
(\sigma_{ij}\bx)_l=x_l-(x_i-x_j)(\delta_{il}-\delta_{jl}).
\end{equation}
It is easy to see that $x_l$ remains unchanged for $l\neq i,j$, that $(\sigma_{ij}\bx)_i=x_j$ and that $(\sigma_{ij}\bx)_j=x_i$. Therefore, the group generated by the reflections along the elements of $A$ with composition as the group operation is the symmetric group $S_N$.

In view of this property of $A$, we see that we can obtain any root from at most two reflections of any other root, as shown below. Consider an arbitrary root $\balpha_{ij}$ and apply to it the reflection $\sigma_{mj}$, with $m$ arbitrary. This reflection exchanges the $j$-th and the $m$-th components of $\balpha_{ij}$, leaving us with $\balpha_{im}$. If we reflect once more using $\sigma_{il}$, with $l$ arbitrary, we obtain $\balpha_{lm}$ as desired. Since $k$ is invariant under any of these reflections, one obtains
\[k(\balpha_{ij})=k(\sigma_{mj}\sigma_{il}\balpha_{ij})=k(\balpha_{lm})\]
in general, and therefore it can be concluded that $k$ is independent of its argument, so it is a single parameter.

\section{Multivariate Special Functions}\label{MVSFns}

Let us summarize four common families of symmetric polynomials, keeping in mind that we will consider polynomials of $N$ variables. Apart from the monomial symmetric functions and the Jack functions (introduced in Section \ref{main1theorem}), we introduce the elementary symmetric and Schur functions.

For an integer $0\leq n\leq N$, the elementary symmetric function $e_n(\bx)$ is given by
\begin{equation}\label{edefinition}
e_n(\bx)=\sum_{1\leq i_1<\cdots< i_n\leq N}\prod_{j=1}^nx_{i_j}.
\end{equation}
For example, $e_0(\bx)=1$, $e_1(\bx)=\sum_{i=1}^Nx_i$, $e_2(\bx)=\sum_{1\leq i<j \leq N}x_ix_j$ and $e_N(\bx)=\prod_{i=1}^N x_i$. When the subscript of $e$ is a partition, it is given by
\begin{equation}
e_{\tau}(\bx)=\prod_{i=1}^{l(\tau)}e_{\tau_i}(\bx).
\end{equation}

The Schur function $s_{\tau}(\bx)$ is given by the Jacobi-Trudi formula \cite{fulton, macdonald}
\begin{equation}
s_{\tau}(\bx)=\frac{\det_{1\leq i,j\leq N}[x_j^{\tau_i+N-i}]}{\det_{1\leq i,j\leq N}[x_j^{N-i}]}.
\end{equation}

%3- Jack functions, normalizations, c and c prime and relationship with e and m

There is one Jack function for each partition $\tau$, which is unique up to a normalization constant. We will use the following two normalizations: the P normalization, which is defined by the following linear combination of the monomial symmetric functions $m_\lambda(\bx)$,
\begin{equation}\label{JackP}
\PP{\tau}{\alpha}(\bx)=\sum_{\substack{\lambda:\lambda\leq\tau\cr |\lambda|=|\tau|}}u_{\tau\lambda}(\alpha)m_{\lambda}(\bx),
\end{equation} 
and the C normalization, which is defined so that
\begin{equation}
\sum_{n=0}^\infty\frac{1}{n!}\sum_{\substack{\tau:l(\tau)\leq N\cr |\tau|=n}}\CC{\tau}{\alpha}(\bx)=\exp (x_1+x_2+\ldots+x_N)\label{JackCexp}
\end{equation}
for all $\alpha$.
The partition-indexed matrix $u_{\tau\lambda}(\alpha)$, as shown in \cite{macdonald}, is an upper triangular matrix whose diagonal entries are equal to one, and its non-diagonal entries are sums of ratios of the form $(a\alpha + b)/(c\alpha + d)$, where $a$, $b$, $c$ and $d$ are non-negative integers. The Jack functions reduce to the monomial, elementary or Schur functions depending on the the parameter $\alpha$, and they are expressed in terms of $m_{\lambda}(\bx)$ in the form of \eref{JackP} as per Table~\ref{Jacktable} \cite{macdonald}.
\Table{\label{Jacktable}The Jack function $\PP{\tau}{\alpha}(\bx)$ and its particular cases. Note that the partition indexing $e_{\tau^\prime}(\bx)$ is a conjugate partition. All matrices are upper triangular with diagonal entries equal to one. The matrix $K_{\tau\lambda}$ is called the Kostka matrix.}
\br
Parameter	&Function				&Matrix					&Function name\\
\mr
$\alpha$		&$\PP{\tau}{\alpha}(\bx)$	&$u_{\tau\lambda}(\alpha)$	&Jack\\
0			&$e_{\tau^\prime}(\bx)$	&$a_{\tau\lambda}$			&Elementary symmetric\\
1			&$s_{\tau}(\bx)$		&$K_{\tau\lambda}$			&Schur\\
$\infty$		&$m_\tau(\bx)$			&$\delta_{\tau\lambda}$		&Monomial symmetric\\
\br
\endTable

Using Jack functions, one can define the generalized hypergeometric function $\FF{\alpha}\left(\bx,\by\right)$ as follows \cite{bakerforrester97,forrester10},
\begin{equation}\label{GHF}
\FF{\alpha}(\bx,\by)=\sum_{n=0}^\infty\frac{1}{n!}\sum_{\substack{\tau:l(\tau)\leq N\cr |\tau|=n}}\frac{\CC{\tau}{\alpha}(\bx)\CC{\tau}{\alpha}(\by)}{\CC{\tau}{\alpha}(\bone)},
\end{equation}
with $\bone$ as given in \eref{vectorone}. For the general definition of the generalized hypergeometric function $\GHG{p}{q}{\alpha}$, see \cite{bakerforrester97}. It is desirable to express the above in terms of the P-normalized Jack functions so that we can use the content of Table~\ref{Jacktable}. For this purpose, we require the generalized Pochhammer symbol $(a)_\tau^{(\alpha)}$, which is defined as the product 
\begin{equation}\label{GPS}
(a)_\tau^{(\alpha)}=\prod_{i=1}^{l(\tau)}\frac{\Gamma(a-(i-1)/\alpha+\tau_i)}{\Gamma(a-(i-1)/\alpha)},
\end{equation}
as well as the functions
%\numparts
\begin{eqnarray}
c_\tau(\alpha)=\prod_{(i,j)\in \tau}(\alpha(\tau_i-j)+\tau_j^\prime-i+1),\nonumber\\
c_\tau^\prime(\alpha)=\prod_{(i,j)\in \tau}(\alpha(\tau_i-j+1)+\tau_j^\prime-i).\label{hooks}
\end{eqnarray}
%\endnumparts
With these definitions, one can write \cite{DumitriuEdelman07}
\begin{eqnarray}
\CC{\tau}{\alpha}(\bx)=\frac{\alpha^{|\tau|}|\tau|!}{c_\tau^\prime(\alpha)}\PP{\tau}{\alpha}(\bx),&\quad&\PP{\tau}{\alpha}(\bone)=\frac{\alpha^{|\tau|}(N/\alpha)_\tau^{(\alpha)}}{c_\tau(\alpha)}.\label{CPconversion}
\end{eqnarray}
Insertion of the above in \eref{GHF} yields, with $\alpha=1/k$,
\begin{equation}\label{RadDunklKerJackP}
\FF{1/k}(\bx,\by)=\sum_{n=0}^\infty\sum_{\substack{\tau:l(\tau)\leq N\cr |\tau|=n}}\frac{c_\tau (1/k)}{c_\tau^\prime (1/k)}\frac{\PP{\tau}{1/k}(\bx)\PP{\tau}{1/k}(\by)}{(kN)_\tau^{(1/k)}}.
\end{equation}

\section{Dunkl Operators, Dunkl Heat Equation and TPD}\label{dunklstuff}

%%%%%%%%%%%%%%%%%%%%%%%%%%%%%%%%%%%%%%%%%%%%%%%%

%%%%%%%%%%%%%%%%%%%%%%%%%%%%%%%%%%%%%%%%%%%%%%%%

Dunkl operators are given by \cite{Dunkl89},
\begin{equation}
T_{\bxi} f(\bx)=\partial_{\bxi}f(\bx)+\sum_{\balpha \in R_+}k(\balpha)\frac{f(\bx)-f(\sigma_{\balpha} \bx)}{\balpha\cdot\bx} \balpha\cdot\bxi\label{DunklDefinition}.
\end{equation}
The first term is a derivative along the vector $\bxi$. Each of the summands in the second term is proportional to the odd part of $f$ along the vector $\balpha$.

Dunkl defined this operator in order to study multivariate orthogonal polynomials and special functions related to reflection groups. Once a root system and multiplicity function are chosen, the operator $T_{\bxi}$ %, for which we introduce the notation $T_i=T_{\be_i}$ 
commutes with $T_{\bxi^\prime}$ and therefore has many of the characteristics of partial derivatives. Hereafter we will use the notation $T_i=T_{\be_i}$ as in Section~\ref{review1}.

In \cite{rosler98}, R\"{o}sler found the TPD $p_k(t,\by|\bx)$ that solves \eref{dunklheatequation} as its Kolmogorov backward equation, with $T_i$ given by \eref{DunklDefinition}. The generalized heat equation \eref{dunklheatequation}, or Dunkl heat equation, is given explicitly by \cite{Dunkl89}
\begin{eqnarray}
\fl\frac{\partial}{\partial t}p_k(t,\by|\bx)=&&\frac{1}{2}\Delta^{(x)} p_k(t,\by|\bx)+\!\!\sum_{\balpha\in R_+}k(\balpha)\frac{\partial_{\balpha} p_k(t,\by|\bx)}{\balpha\cdot\bx}\nonumber\\
&&\qquad-\sum_{\balpha\in R_+}\!\!k(\balpha)\frac{\alpha^2}{2}\frac{p_k(t,\by|\bx)-p_k(t,\by|\sigma_{\balpha} \bx)}{(\balpha\cdot\bx)^2}.\label{DunklHeat}
\end{eqnarray}

The calculation of $p_k(t,\by|\bx)$ can be accomplished using the Dunkl transform, which requires both the Dunkl kernel defined in Section~\ref{review2} and the weight function
\begin{equation}
w_k(\bx)=\prod_{\balpha\in R}|\balpha\cdot\bx|^{k(\balpha)}.
\end{equation}
For fixed $R$ and $k$ and functions $f\in L^1(\RR^N,w_k)$, the set of integrable functions with respect to the weight function $w_k(\bx)$, the Dunkl transform is defined by the equation \cite{dunkl92}
\begin{equation}
\hat{f}_k(\bxi)=\frac{1}{c_k}\int_{\RR^N}f(\bx)E_k(-\ii\bxi,\bx)w_k(\bx)\ud^N x,\label{dunkltransform}
\end{equation}
and
\begin{equation}
c_k=\int_{\RR^N}\rme^{-\xi^2/2}w_k(\bxi)\ud^N \xi
\end{equation}
is a normalization constant obtained from a Selberg integral \cite{Mehta04}. This transform has many of the properties of the Fourier transform, and it can be used to solve equations involving Dunkl operators in the same way the Fourier transform is used to solve differential equations.

%Due to the properties of the Dunkl kernel it is easy to see that
%\begin{equation}
%T_j^{\bxi}\hat{f}_k(\bxi)=\frac{1}{c_k}\int_{\RR^N}[-ix_jf(\bx)]E_k(-i\bxi,\bx)w_k(\bx)\ud^N x.\label{property1}
%\end{equation}
%On the other hand, since
%\begin{equation}
%\frac{\partial}{\partial x_j}w_k(\bx)=2\sum_{\balpha\in R_+}\alpha_j\frac{k(\balpha)}{\balpha\cdot\bx}w_k(\bx),
%\end{equation}
%a direct calculation shows that
%\begin{equation}
%\int_{\RR^N}[T_jf(\bx)]g(\bx)w_k(\bx)\ud^N x=-\int_{\RR^N}f(\bx)[T_jg(\bx)]w_k(\bx)\ud^N x\label{innerproductantisymmetry}
%\end{equation}
%for $f$ and $g$ integrable with respect to $w_k$. Using Eq.~\eref{innerproductantisymmetry}, one can prove directly that 
%\begin{equation}
%\frac{1}{c_k}\int_{\RR^N}[T_j^{\bx}f(\bx)]E_k(-i\bxi,\bx)w_k(\bx)\ud^N x=i\xi_j \hat{f}_k(\bxi).\label{property2}
%\end{equation}
%The inverse Dunkl transform is defined as
%\begin{equation}
%f(\bx)=\frac{1}{c_k}\int_{\RR^N}\hat{f}_k(\bxi)E_k(i\bxi,\bx)w_k(\bxi)\ud^N \xi,\label{inversetransform}
%\end{equation}
%and it is known that the inverse transform of the transform of $f$ is equal to $f$ almost everywhere. In view of the identity
%\begin{equation}
%\frac{1}{(2\pi)^{N/2}}\int_{\RR^N}\hat{f}(\bxi)\rme^{i\bx\cdot\bxi}\rme^{i\by\cdot\bxi}\ud ^N \xi=f(\bx+\by),\label{fouriertranslation}
%\end{equation}
%where $\hat{f}$ is the $N$-dimensional Fourier transform of $f$, 

The generalized Dunkl translation $\tau_{\by}$ is defined as 
\begin{equation}
\tau_{\by}f(\bx)=\frac{1}{c_k}\int_{\RR^N}\hat{f}_k(\bxi)E_k(\ii\bx,\bxi)E_k(\ii\by,\bxi)w_k(\bxi)\ud^N \xi,\label{generalizedtranslation}
\end{equation}
which reduces to a regular translation when $k=0$. Because \eref{generalizedtranslation} does not change when $\bx$ and $\by$ are exchanged, we see that $\tau_{\by}f(\bx)=\tau_{\bx}f(\by)$.

Using the Dunkl transform on \eref{DunklHeat}, one obtains its Green function centred at the origin. Using the translation \eref{generalizedtranslation} yields the complete Green function. Let us define the sum of $k(\balpha)$ over the positive subsystem as
\begin{equation}
\gamma=\sum_{\balpha\in R_+}k(\balpha).
\end{equation}
Then, the Green function of \eref{DunklHeat} is given by

%One could think that the calculation of the Green function is finished by giving $g_k$ a displacement $\by$ to obtain 
%\[g_k(\by-\bx,t)\sim \int_{\RR^N}\rme^{-t\xi^2/2}\rme^{-i\bxi\cdot\bx}\rme^{i\bxi\cdot\by}\ud^N \xi.\]
%However, this function clearly fails to solve Eq.~\eref{DunklHeat}. Instead, we must use the generalized Dunkl translation Eq.~\eref{generalizedtranslation},
%\begin{equation}
%\tau_{-\bx}g_k(\by,t)=\frac{1}{c_k}\int_{\RR^N}A\rme^{-t\xi^2/2}E_k(-i\bx,\bxi)E_k(i\by,\bxi)w_k(\bxi)\ud^N \xi,
%\end{equation}
%which automatically solves Eq.~\eref{DunklHeat} on both $\bx$ and $\by$. With the substitution $\bxi^\prime =\sqrt t \bxi$ and using Eq.~\eref{usefulintegral} we obtain the Green function
%\begin{equation}
%\Gamma_k(t,\by|\bx)=\tau_{-\bx}g_k(\by,t)=\frac{A\rme^{-(x^2+y^2)/2t}}{t^{N/2+\gamma}}E_k\left(\frac{\bx}{\sqrt t},\frac{\by}{\sqrt t}\right).
%\end{equation}
%In other words, the Dunkl transform of a Gaussian centered at the origin is a Gaussian, but in order to displace the Gaussian correctly and obtain the Green function above, one requires the generalized translation in Eq.~\eref{generalizedtranslation}.

\begin{equation}
\Gamma_k(t,\by|\bx)=\frac{\rme^{-(x^2+y^2)/2t}}{c_k t^{N/2+\gamma}}E_k\left(\frac{\bx}{\sqrt t},\frac{\by}{\sqrt t}\right).
\end{equation}
It is known that this Green function is normalized with respect to $w_k(\bx)$. That is,
\begin{eqnarray}
\int_{\RR^N}\Gamma_k(t,\by|\bx)w_k(\by)\ud^N y&=&1.
\end{eqnarray}
Therefore, we recognize the TPD of the Dunkl process to be
%Note that, since we are interested in calculating a TPD based on $\Gamma_k$, it must be normalized. In particular, we must impose the condition that, for an arbitrary starting point $\bx$ and time $t$,
%\begin{equation}
%\int_{\RR^N}p_k(t,\by|\bx)\ud^Ny=1.\label{probnorm}
%\end{equation}
%As opposed to the case of the $N$-dimensional Brownian motion, this TPD is not given by $\Gamma_k$. For example, if we try to normalize $\Gamma_k$ directly, we obtain for the case $\bx=\bzero$
%\begin{equation}
%1=\int_{\RR^N}\Gamma_k(t,\by|\bzero)\ud^N y%\nonumber\\
%=\frac{A}{t^{N/2+\gamma}}\int_{\RR^N}\rme^{-y^2/2t}\ud^N y=\frac{A(2\pi)^{N/2}}{t^\gamma}.
%\end{equation}
%This leads to $A=t^\gamma/(2\pi)^{N/2}$, which is a clear contradiction because of the presence of $t$. Therefore, we realize that the normalization of $\Gamma_k$ must be calculated with respect to $w_k$ in order to obtain a consistent result for $A$:
%\begin{eqnarray}
%1&=&\int_{\RR^N}\Gamma_k(t,\by|\bx)w_k(\by)\ud^N y\nonumber\\
%&=&A\frac{\rme^{-x^2/2t}}{t^{N/2+\gamma}}\int_{\RR^N}\rme^{-y^2/2t}E_k\left(\frac{\bx}{\sqrt t},\frac{\by}{\sqrt t}\right)w_k(\by)\ud^N y=Ac_k.
%\end{eqnarray}
%For this result, we have used Eq.~\eref{usefulintegral}. This yields $A=c_k^{-1}$ and
%\begin{equation}
%\Gamma_k(t,\by|\bx)=\frac{\rme^{-(x^2+y^2)/2t}}{c_k t^{\gamma+N/2}}E_k\left(\frac{\bx}{\sqrt t},\frac{\by}{\sqrt t}\right).
%\end{equation}
%In view of Eq.~\eref{probnorm}, we conclude that
\begin{eqnarray}
p_k(t,\by|\bx)&=&w_k(\by)\Gamma_k(t,\by|\bx)=w_k(\by)V_kp_0(t,\by|\bx)\nonumber\\
&=&w_k\left(\frac{\by}{\sqrt{t}}\right)\frac{\rme^{-(x^2+y^2)/2t}}{c_k t^{N/2}}E_k\left(\frac{\bx}{\sqrt{t}},\frac{\by}{\sqrt{t}}\right).\label{TPDk}
\end{eqnarray}
The details of the calculation leading to this TPD can be found in \cite{rosler98}. The factor $w_k\left(\by/\sqrt t\right)$ can be interpreted as follows: in \eref{DunklHeat}, the denominators in the second and third terms on the RHS denote a repulsion from the planes defined by $\balpha\cdot\bx=0$, $\balpha\in R$, and therefore, the probability of any process finishing at any point such that $\balpha\cdot\by=0$ should be zero. Indeed, $w_k(\by)=\prod_{\balpha\in R}|\balpha\cdot\by|^{k(\balpha)}$ cancels whenever $\balpha\cdot\by=0$ for any $\balpha\in R$, which reflects this fact.

For the purposes of this work, we are interested in $p_k(t,\by|\bx)$ for the root system of type $A$. In this case, the Dunkl operator in the direction $\bxi$ is given by
\begin{equation}\label{DunklAGen}
T_{\bxi} f(\bx)=\partial_{\bxi}f(\bx)+k\sum_{1\leq i<j\leq N}\frac{f(\bx)-f(\sigma_{ij} \bx)}{x_j-x_i} (\xi_j-\xi_i),
\end{equation}
from which \eref{DunklDefinitionA} follows. Let us calculate the quantities that form $p_k$ in this particular case. First, the weight function $w_k$ becomes proportional to the Vandermonde determinant $h_N(\bx)$,
\begin{equation}
w_k(\bx)=\prod_{1\leq j<l\leq N}\left|\frac{1}{\sqrt{2}}(\be_l-\be_j)\cdot\bx\right|^{2k}=\frac{|h_N(\bx)|^{2k}}{2^{kN(N-1)/2}}.
\end{equation}
The normalization constant $c_k$ is then
\begin{eqnarray}
\fl c_k&=&\frac{1}{2^{kN(N-1)/2}}\int_{\RR^N}\rme^{-x^2/2}|h_N(\bx)|^{2k}\ud^N x=\frac{(2\pi)^{N/2}}{2^{kN(N-1)/2}}\prod_{j=1}^N\frac{\Gamma(1+jk)}{\Gamma(1+k)},
\end{eqnarray}
where we have used a particular case of the Selberg integral (equation (17.6.7) in \cite{Mehta04}). Inserting the above in \eref{TPDk} yields \eref{TPDkA}.

\section{Extema and Maximum value of $F_N(\bv,t)$}\label{extrema}

Here, we prove two lemmas necessary to complete the proof of Theorem~\ref{lasttheorem}. The first concerns the location of the extrema of $F_N(\bv,t)$.

\begin{lemma}\label{lemmahermiteroots}
The extrema of the function
\begin{equation}
F_N(\bv,t)=\frac{N}{2}(N-1)(1-\log t)-\sum_{j=1}^Nj\log j+2\log|h_N(\bv)|-\frac{v^2}{2t}
\end{equation}
are located at $\bv=\sqrt{2t}\bz_N$ or any of its permutations, where $\bz_N$ is the vector of roots of the Hermite polynomial $H_N(x)$, given by \eref{zroots} and \eref{hermitep} respectively. Furthermore, all extrema are local maxima.
\end{lemma}
\begin{proof}
We first calculate the first-order partial derivatives of $F_N(\bv,t)$ relative to $\bv$ and equate them to zero.
\begin{equation}
\frac{\partial}{\partial v_i}F_N(\bv,t)=\sum_{\substack{j:j\neq i\cr j=1}}^N\frac{2}{v_i-v_j}-\frac{v_i}{t}=0
\end{equation}
Hence, the extrema of $F_N(\bv,t)$ must obey the relation
\begin{equation}
v_i=\sum_{\substack{j:j\neq i\cr j=1}}^N\frac{2t}{v_i-v_j}.\label{extcond}
\end{equation}
The second order derivatives of $F_N(\bv,t)$ are
\begin{eqnarray}
\frac{\partial^2}{\partial v_j\partial v_i}\left[2\log|h_N(\bv)|-\frac{v^2}{2t}\right]&=&\frac{\partial}{\partial v_j}\left[\sum_{l:l\neq i}\frac{2}{v_i-v_l}-\frac{v_i}{t}\right]\nonumber\\
&=&
\left\{
\begin{array}{cl}
-\sum_{l:l\neq i}\frac{2}{(v_i-v_l)^2}-\frac{1}{t} & \textrm{if }i=j,\\
\frac{2}{(v_i-v_j)^2} & \textrm{if }i\neq j.
\end{array}
\right.\label{secondder}
\end{eqnarray}
The matrix formed by the $N\times N$ second order derivatives above is negative definite for all vectors $\bv$ with non-repeating components. To show this, we consider an arbitrary real vector $\bu$ and calculate the quadratic form associated to \eref{secondder}.
\begin{equation}
\sum_{1\leq i,j\leq N}u_i\frac{\partial^2F_N(\bv,t)}{\partial v_j\partial v_i}u_j=-\frac{1}{t}\sum_{i=1}^Nu_i^2-2\sum_{1\leq i<j\leq N}\frac{(u_i-u_j)^2}{(v_i-v_j)^2}\leq 0\label{quadraticform}
\end{equation}
Here, the equality holds only when all the ${u_i}$ are equal to zero. Hence, all extrema given by \eref{extcond} are maxima.

Let us focus on the location of the extrema. Applying the scaling $\bv=\sqrt{2t}\bz$ to \eref{extcond} we obtain
\begin{equation}
z_{i}=\sum_{\substack{j:j\neq i\cr j=1}}^N\frac{1}{z_{i}-z_{j}},\label{eqforz}
\end{equation}
for all $i=1\ldots,N$. Therefore, it suffices to solve the above equation for $\bz$ to find the location of the extrema of $F_N(\bv,t)$.
Note that, given $\bz$, any of its permutations solve \eref{eqforz}:
\begin{equation}
z_{\rho(i)}=\sum_{\substack{j:j\neq i\cr j=1}}^N\frac{1}{z_{\rho(i)}-z_{\rho(j)}}
\end{equation}
for any $\rho\in S_N$. Now we prove that \eref{eqforz} implies that $\{z_i\}_{i=1,\ldots,N}$ must be the roots of the $N$-th Hermite polynomial. Let us multiply \eref{eqforz} by $\prod_{\substack{l:l\neq i\cr l=1}}^N(z_i-z_j)$.
\begin{equation}
z_{i}\prod_{\substack{l:l\neq i\cr l=1}}^N(z_i-z_l)=\sum_{\substack{j:j\neq i\cr j=1}}^N\prod_{\substack{l:l\neq i,j\cr l=1}}^N(z_i-z_l)\label{eqforzprod}
\end{equation}
Now, we consider a polynomial whose roots are $\{z_i\}_{i=1,\ldots,N}$:
\begin{equation}
p(x)=c \prod_{n=1}^N(x-z_n),
\end{equation}
with $c$ a non-zero constant. The first two derivatives of this polynomial are:
\begin{equation}
p^\prime(x)=\frac{\ud}{\ud x}p(x)=c\sum_{j=1}^N\prod_{\substack{n:n\neq j\cr n=1}}^N(x-z_n)
\end{equation}
and
\begin{equation}
p^{\prime\prime}(x)=\frac{\ud^2}{\ud x^2}p(x)=2c\sum_{1\leq j<l\leq N}\prod_{\substack{n:n\neq j,l\cr n=1}}^N(x-z_n).
\end{equation}
At any of the values $z_i$, $p^{\prime\prime}(x)$ behaves as follows.
\begin{equation}
p^{\prime\prime}(z_i)=2c\sum_{\substack{j:j\neq i\cr j=1}}^N\prod_{\substack{n:n\neq i,j\cr n=1}}^N(z_i-z_n)
\end{equation}
We insert \eref{eqforzprod} to obtain
\begin{equation}
p^{\prime\prime}(z_i)=2cz_i\prod_{\substack{n:n\neq i\cr n=1}}^N(z_i-z_n)=2z_ip^\prime(z_i).\label{hzeros}
\end{equation}
It is known \cite{szego} that the differential relation on the zeros of the polynomial $p(x)$ is only fulfilled by the $N$-th Hermite polynomial. Indeed, it solves the differential equation
\begin{equation}
H^{\prime\prime}_N(x)-2xH^\prime_N(x)+2NH_N(x)=0,
\end{equation}
which reduces to \eref{hzeros} when $x=z_{i,N},$ with $i=1,\ldots,N$ and $z_{i,N}$ is the $i$-th root of $H_N(x)$. Hence, $p(x)\propto H_N(x),$ and $z_i=z_{i,N}$.
\end{proof}

In the second lemma, we find the maximum value of $F_N(\bv,t)$.
\begin{lemma}
The maximum value of $F_N(\bv,t)$ is zero.
\end{lemma}
\begin{proof}
By Lemma~\ref{lemmahermiteroots}, the maximum value of $F_N(\bv,t)$ is located at $\bv=\sqrt{2t}\bz_N$. This yields
\begin{eqnarray}
\fl F_N(\sqrt{2t}\bz_N,t)&=&\frac{N}{2}(N-1)-z_N^2+2\log|h_N(\bz_N)|+\frac{N}{2}(N-1)\log 2-\sum_{j=2}^Nj\log j,
\end{eqnarray}
which is independent of $t$. Let us focus first on the term
\begin{equation}\label{squaretocompute}
z_N^2=\sum_{j=1}^Nz_{j,N}^2.
\end{equation}
The derivative property of the Hermite polynomials \cite{arfken} allows us to write
\begin{equation}\label{hermitederivativex}
H_N^\prime(x)=2NH_{N-1}(x),
\end{equation}
and using the product representation from Lemma~\ref{lemmahermiteroots} we obtain
\begin{equation}
\sum_{i=1}^N\prod_{\substack{j:j=1\cr j\neq i}}^N(x-z_{j,N})=N\prod_{j=1}^{N-1}(x-z_{j,N-1}).
\end{equation}
Since the above equation holds for all $x$, we can expand in powers of $x$ and equate the coefficients of the same order. In particular, the coefficients of $x^{N-3}$ form the equation
\begin{equation}
(N-2)\sum_{1\leq i<j\leq N}z_{i,N}z_{j,N}=N\sum_{1\leq i<j\leq N-1}z_{i,N-1}z_{j,N-1},
\end{equation}
and if we denote the double sum on the LHS by $s_{N}$, we obtain
\begin{equation}
s_{N}=\frac{N}{N-2}s_{N-1}.
\end{equation}
Using mathematical induction on this relation we find that
%Now, we can show that $s_{N}=-N(N-1)/4$ by induction. Since $H_2(x)=4x^2-2$, its two roots lie at $x=\pm\sqrt{1/2}$ and therefore $s_{2}=-1/2$. Then, the formula is valid for $N=2$. Assuming it is also valid for $N-1$, we verify that it holds for $N$:
\begin{equation}
s_{N}=-\frac{N(N-1)}{4}.
\end{equation}
Now, since the Hermite polynomials have either odd or even symmetry, the sum of their roots equals zero. Therefore,
\begin{equation}
z_N^2=\left(\sum_{i=1}^Nz_{i,N}\right)^2-2s_{N}=\frac{N(N-1)}{2}.
\end{equation}

Now, let us calculate
\begin{equation}\label{logtocompute}
2\log|h_N(\bz_N)|.
\end{equation}
The main quantity we wish to find is the square of the Vandermonde determinant of the roots of the $N$-th Hermite polynomial. This quantity is known as the discriminant of the Hermite polynomials, and we calculate it here following Szeg\"{o} \cite{szego}. Using 
\begin{equation}\label{hermitederivative}
H_N^\prime(z_{i,N})=\lim_{x\to z_{i,N}}\frac{H_N(x)}{x-z_{i,N}}=2^N\prod_{\substack{j:j\neq i\cr j=1}}^N(z_{i,N}-z_{j,N}),
\end{equation}
we can write
\begin{eqnarray}
(h_N(\bz_N))^2&=&\prod_{1\leq i<j\leq N}(z_{j,N}-z_{i,N})^2=(-1)^{N(N-1)/2}\prod_{1\leq i\neq j\leq N}(z_{j,N}-z_{i,N})\nonumber\\
&=&\frac{(-1)^{N(N-1)/2}}{2^{N^2}}\prod_{i=1}^NH_N^\prime(z_{i,N})\nonumber\\
&=&\frac{(-1)^{N(N-1)/2}N^N}{2^{N(N-1)}}\prod_{i=1}^NH_{N-1}(z_{i,N}).
\end{eqnarray}
The last equality follows from \eref{hermitederivativex}. Let us focus on the last product:
\begin{eqnarray}
\prod_{i=1}^NH_{N-1}(z_{i,N})&=&2^{N(N-1)}\prod_{i=1}^N\prod_{j=1}^{N-1}(z_{i,N}-z_{j,N-1})\nonumber\\
&=&2^{N(N-1)}\prod_{j=1}^{N-1}\prod_{i=1}^{N}(z_{j,N-1}-z_{i,N})=\prod_{j=1}^{N-1}H_N(z_{j,N-1}).
\end{eqnarray}
From the recurrence relation
\begin{equation}
H_N(x)=2xH_{N-1}(x)-2(N-1)H_{N-2}(x)
\end{equation}
we obtain that $H_N(z_{j,N-1})=-2(N-1)H_{N-2}(z_{j,N-1})$, so the product above becomes
\begin{equation}
\prod_{i=1}^NH_{N-1}(z_{i,N})=[-2(N-1)]^{N-1}\prod_{j=1}^{N-1}H_{N-2}(z_{j,N-1}).
\end{equation}
Mathematical induction on the above yields
%\begin{equation}
%H_1(-1/\sqrt 2)H_1(1/\sqrt{2})=-2,
%\end{equation}
%for which the general formula
%\begin{equation}\label{discriminant1}
%\prod_{i=1}^NH_{N-1}(z_{i,N})=(-2)^{N(N-1)/2}\prod_{j=1}^{N-1}j^j
%\end{equation}
%holds. Now, assuming that Eq.~\eref{discriminant1} holds for $N-1$, we verify that it holds for $N$:
%\begin{eqnarray}
%\prod_{i=1}^NH_{N-1}(z_{i,N})&=&[-2(N-1)]^{N-1}\prod_{j=1}^{N-1}H_{N-2}(z_{j,N-1})\nonumber\\
%&=&[-2(N-1)]^{N-1}(-2)^{(N-1)(N-2)/2}\prod_{j=1}^{N-2}j^j\nonumber\\
%&=&(-2)^{N(N-1)/2}\prod_{j=1}^{N-1}j^j.
%\end{eqnarray}
\begin{eqnarray}
\prod_{i=1}^NH_{N-1}(z_{i,N})&=&(-2)^{N(N-1)/2}\prod_{j=1}^{N-1}j^j.
\end{eqnarray}
Therefore,
\begin{eqnarray}
(h_N(\bz_N))^2&=&\frac{(-1)^{N(N-1)/2}N^N}{2^{N(N-1)}}\prod_{i=1}^NH_{N-1}(z_{i,N})=\frac{1}{2^{N(N-1)/2}}\prod_{j=1}^{N}j^j.
\end{eqnarray}
The logarithm of the above is
\begin{equation}
2\log|h_N(\bz_N)|=\sum_{j=1}^Nj\log j-\frac{N}{2}(N-1)\log 2,
\end{equation}
and thus the maximum value of $F_N(\bv,t)$ is zero for all its extrema, due to the fact that \eref{squaretocompute} and \eref{logtocompute} do not change when the roots $\bz_N$ are permuted.
\end{proof}

\bibliography{biblio}

\end{document}